\documentclass[USenglish,a4paper,11pt]{article}
\pdfoutput=1

\usepackage[T1]{fontenc}
\usepackage[margin=1in]{geometry}
\usepackage{amsmath, amssymb, amsfonts, amsthm, mathtools}
\usepackage{bbm, dsfont, bm}
\usepackage{graphicx}
\usepackage{tikz}
\usetikzlibrary{shapes.geometric, fit, arrows.meta,decorations.pathreplacing}
\usepackage[dvipsnames]{xcolor}
\usepackage{url}
\usepackage{hyperref}
\usepackage{pgfplots}
\usepackage{tikz-cd}
\usetikzlibrary{positioning}
\usepackage{cleveref}
\usepackage{enumerate, enumitem}
\usepackage{algorithm, algpseudocode} 
\usepackage{subfigure}
\usepackage{tcolorbox}
\usepackage{enumitem}
\usepackage{tabularx}
\usepackage{makecell}
\usepackage{rotating}
\tcbuselibrary{skins,breakable}
\usepackage{pifont} 
\usepackage{framed}
\usepackage[backref=true,style=alphabetic,citestyle=alphabetic,minalphanames=3,maxalphanames=4,maxbibnames=99,maxcitenames=99]{biblatex}

\usepackage{epigraph}
\definecolor{DarkRed}{rgb}{0.5,0.1,0.1}
\definecolor{DarkBlue}{rgb}{0.1,0.1,0.5}
\definecolor{ForestGreen}{rgb}{0.1333,0.5451,0.1333}
\definecolor{TodoBg}{HTML}{FFF4CC} 
\definecolor{TodoFrame}{HTML}{E09F3E} 
\definecolor{TodoTitle}{HTML}{9A3412} 
\colorlet{shadecolor}{orange!15}

\hypersetup{
    linktocpage=true,
    pagebackref=true,
    colorlinks=true,
    linkcolor=BrickRed,
    citecolor=MidnightBlue,
    filecolor = blue,
    urlcolor=OliveGreen,
    bookmarks=true,
    bookmarksopen=true,
    bookmarksnumbered=true
}

\usepackage{float}
\newfloat{markovchain}{htbp}{lom}
\floatname{markovchain}{Markov Chain}

\crefname{algorithm}{algorithm}{algorithms}
\Crefname{algorithm}{Algorithm}{Algorithms}
\algrenewcommand\algorithmicrequire{\textbf{Input:}}
\algrenewcommand\algorithmicensure{\textbf{Output:}} 

\crefname{markovchain}{markov chain}{markov chains}
\Crefname{markovchain}{Markov Chain}{Markov Chains}

\theoremstyle{plain}
\newtheorem{theorem}{Theorem}[section]
\newtheorem{lemma}[theorem]{Lemma}
\newtheorem{proposition}[theorem]{Proposition}

\newtheorem{fact}[theorem]{Fact}

\newtheorem{remark}[theorem]{Remark}

\theoremstyle{definition}
\newtheorem{definition}[theorem]{Definition}

\newtheorem*{definition*}{Definition}
\newtheorem*{theorem*}{Theorem}
\newtheorem*{lemma*}{Lemma}
\newtheorem*{proposition*}{Proposition}
\newtheorem*{corollary*}{Corollary}
\newtheorem*{remark*}{Remark}
\newtheorem*{claim*}{Claim}
\newtheorem*{conjecture*}{Conjecture}
\newtheorem*{openq*}{Open Question}
\newtheorem*{question*}{Question}

\newtheoremstyle{restate}{}{}{\itshape}{}{\bfseries}{.}{.5em}{\thmnote{#3}}
\theoremstyle{restate}
\newtheorem*{rtheorem}{Theorem}
\newtheorem*{rproposition}{Proposition}
\newtheorem*{rlemma}{Lemma}

\newcommand{\E}{\mathbb{E}}

\newcommand{\norm}[1]{\left\lVert #1 \right\rVert}

\DeclarePairedDelimiter\bra{\langle}{\rvert} 
\DeclarePairedDelimiter\ket{\lvert}{\rangle} 
\DeclarePairedDelimiterX\braket[2]{\langle}{\rangle}{#1\,\delimsize\vert\,\mathopen{}#2}
\DeclarePairedDelimiterX\brasket[3]{\langle}{\rangle}{#1\,\delimsize\vert\,\mathopen{}#2\delimsize\vert\,\mathopen{}#3}

\newcommand{\heavy}{\textnormal{\textsc{Heavy}}}
\newcommand{\new}{\textnormal{\textsc{New}}}
\newcommand{\margds}{\mcl{L}_{\rho}}
\newcommand{\Cut}{\mathrm{Cut}}
\newcommand{\Cycle}{\mathrm{Cycle}}
\newcommand{\mcl}[1]{\mathcal{#1}}

\newcommand{\abs}[1]{\left|#1\right|}

\newcommand{\Span}{\mathrm{span}}
\newcommand{\toh}{\widetilde{O}}
\newcommand{\poly}{\mathrm{poly}}
\newcommand{\oh}{O}
\newcommand{\wtd}[1]{\widetilde{#1}}
\newcommand{\reff}{R_{\rm eff}}
\newcommand{\tdreff}{\widetilde{R}_{\rm eff}}
\newcommand{\orcl}{\mcl{O}_G}

\newlist{todolist}{itemize}{2} 
\setlist[todolist]{label=$\square$} 

\allowdisplaybreaks

\title{Quantum Sampling of Random Spanning Trees \\ via Amortized Data Structures}

\author{Yassine Hamoudi\thanks{\texttt{yassine.hamoudi@labri.fr}}\qquad\; Adrian Tanasa\thanks{\texttt{adrian.tanasa@labri.fr}}\qquad\; Shrinidhi Teganahally Sridhara\thanks{\texttt{shrinidhi.teganahally-sridhara@labri.fr}} \\ [12pt] LaBRI, Univ. Bordeaux, CNRS UMR 5800, F-33400 Talence, France}

\date{\today}

\begin{document}
\maketitle

\begin{abstract}
  
We present a quantum algorithm that generates a uniform superposition over the spanning trees of a graph -- also known as a \emph{q-sample} -- using only a sublinear number of queries to the graph. 
This goes beyond previous work that focuses on generating \emph{classical} samples, including a recent quantum algorithm by Apers, Gao, Ji, and Liu~\cite{AGJL25c}. For an $n$-vertex, $m$-edge graph, our algorithm admits a pre-processing step of $\toh(\sqrt{mn} + m^{1-\delta})$, after which each q-sample can be generated in~$\toh(n^{1+2\delta})$ for any~$\delta$. Alternatively, $k$ independent q-samples can be generated at a total cost of~$\toh(\sqrt{kmn})$. We also prove a matching lower bound up to logarithmic factors, showing that our algorithm is essentially optimal. Further tradeoffs are given in the paper. In comparison, the optimal classical algorithms~\cite{ALV22c} need an $\toh(m)$ pre-processing step, after which each sample costs~$\toh(n)$ operations. These results yield faster quantum walk-based algorithms for counting spanning trees or finding a marked one.

Our result is obtained via quantum walk sampling over a sequence of slowly-changing Markov chains. Each chain is an \emph{isotropized up-down walk} that mixes rapidly to the spanning tree distribution of the input graph. A key ingredient is an amortized data structure supporting fast implementation of the associated quantum walk operators throughout the sequence. This data structure maintains access to a spectral sparsifier and a leverage score sampler that evolve with the underlying graph.

\end{abstract}

\section{Introduction}
Sampling and counting discrete structures in graphs is a fundamental task in randomized algorithms, with examples including independent sets, matchings, cuts, paths, and spanning trees. Random spanning trees are among the most extensively studied examples: since Kirchhoff's celebrated matrix-tree theorem~\cite{Kir47j} they have been closely tied to determinant identities, electrical networks, and graph Laplacians, see for example~\cite{DS84b,LP17ba}. Moreover, random spanning trees  form a canonical class of strongly Rayleigh measures, and hence exhibit strong negative-dependence properties~\cite{BBL09j}.

For a weighted connected graph $G(V,E,w)$ with $n$ vertices and $m$ edges, an important algorithmic problem is to sample one spanning tree from the weighted distribution in which each tree $T$ is chosen with probability proportional to $\prod_{e\in T}w(e)$. This task has been studied from several complementary perspectives.

A first family of algorithms thus exploits the matrix-tree theorem in a direct manner: Laplacian cofactors give exact counting information that can be used to guide the sampling process. This determinant-based approach led to polynomial-time samplers by Guenoche~\cite{Gue83j} and Kulkarni~\cite{Kul90j}, and was later refined by Colbourn, Myrvold, and Neufeld to run in $O(n^\omega)$ arithmetic operations~\cite{CMN96j} where $\omega \approx 2.37$ is the matrix multiplication exponent.

Another standard family of such algorithms gives exact samplers using random walks. The Aldous--Broder algorithm samples a uniform spanning tree from the first-entrance tree of a random walk~\cite{Ald90j,Bro89c}, and Wilson's loop-erased random-walk algorithm gave a faster and more practical exact sampler~\cite{Wil96c}. Subsequent work accelerated this random-walk viewpoint using Laplacian solvers, electrical flows, and Schur complements~\cite{KM09c,MST15c}, culminating in Schild's $m^{1+o(1)}$ time algorithm which is an almost-linear-time exact sampler~\cite{Sch18c}.

In parallel, determinant-preserving sparsification algorithms give approximate samplers running in $\toh(n^2\cdot \epsilon^{-2})$ time~\cite{DPPR20j}. A more recent direction studies approximate sampling through high-dimensional expansion and down-up walks; this viewpoint gives sharp mixing analyses for natural walks on spanning trees and led to near-linear-time approximate samplers, with runtime $O(m\log^2 m)$ in the work of Anari, Liu, Oveis Gharan, Vinzant, and Vuong~\cite{ALG+21c}.

When additional information on the target distribution is available, sampling can be accelerated further. A powerful example is domain sparsification for strongly Rayleigh distributions~\cite{AD20c,ADVY22c,ALV22c}: using approximate element marginals, the sampling task is reduced to a sampling task over a much smaller effective domain. For spanning trees, this gives an $\toh(n)$-step sampler after $\toh(m)$ marginal pre-processing~\cite{ALV22c}.

Most recently, Apers, Gao, Ji and Liu~\cite{AGJL25c} obtained a quantum speedup for the classical approximate sampling task: in the adjacency-list model, their algorithm outputs the classical description of a random spanning tree using $\toh(\sqrt{mn})$ queries and time, matching a corresponding quantum query lower bound up to polylogarithmic factors.

\medskip

Our paper studies the coherent analogue of this task. Given sparse adjacency-list oracle access to a weighted connected graph $G=(V,E,w)$, we aim to prepare the quantum-sample (q-sample) of its spanning-tree distribution,
\[
    \ket{\pi}
    =
    \sum_T \sqrt{\pi(T)}\ket T,
    \qquad
    \pi(T)
    =
    \frac{\prod_{e\in T}w(e)}{\sum_{T'}\prod_{e\in T'}w(e)},
\]
where $\ket T$ stores the lexicographically ordered list of canonical representation of the weighted edges of $T$. Query complexity counts only calls to the adjacency-list oracle.
 
\medskip
Quantum sampling was put forward by Aharonov and Ta-Shma~\cite{AT07j} as a state-generation task. A reversible classical sampler can be used directly to prepare $\sum_r\ket r\ket{T(r)}$, but not to obtain the desired state $\sum_T\sqrt{\pi(T)}\ket T$, since the random-seed register $\ket{r}$ cannot generally be efficiently erased. They also showed that a generic efficient procedure for turning classical samplers into quantum ones would imply $\mathrm{SZK}\subseteq\mathrm{BQP}$, underscoring that q-sampling is likely a difficult task in general. Among the special-purpose quantum algorithms developed for q-sampling, \emph{quantum simulated annealing} has proven particularly successful~\cite{SBBK08j,WA08j,OBD18j,HW20c,CLL22c} and is among the primitives that underpin our algorithm.

Beyond state generation, q-samples are a recurring primitive across several quantum algorithms. In \emph{learning theory}~\cite{AdW17j}, some classes of functions (such as Parities~\cite{BV97j} or Disjunctive Normal Forms~\cite{BJ99j}) can be learned using super-polynomially fewer quantum samples than classical ones. In \emph{Markov Chain Monte Carlo methods}, quantum walk search frameworks typically start from the coherent encoding of the stationary distribution and then amplify or detect marked states in the chain's state space~\cite{MNRS11j,Ben20d,AGJK20c}; this stationary-q-sample starting point underlies applications such as element distinctness~\cite{Amb07j} and triangle finding~\cite{MSS07j}. This has also led to work on the cost of obtaining such starting states, through quantum mixing, controlled quantum walks, analog quantum dynamics, and seed-set methods~\cite{CLR20j,Ape19c,LS23j}. Finally, coherent encoding of a system's initial state is essential for \emph{simulating its evolution} under quantum dynamics, such as via Hamiltonian simulation or the quantum adiabatic algorithm.

Our amortized data structure viewpoint is also related in spirit to recent quantum algorithms that maintain samplers for slowly changing distributions. Of particular relevance is the work of Bouland, Getachew, Jin, Sidford, and Tian~\cite{BGJ+23c}, who apply such ideas to zero-sum games. There, the evolving distributions are Gibbs distributions, whereas in our setting they are leverage-score distributions along a cooling path.

\subsection{Main result}

The main result of this paper is a sublinear-query quantum algorithm for preparing q-samples of the spanning tree distribution. Beyond preparing a single coherent sample, the algorithm gives a range of preprocessing--sampling tradeoffs and is particularly useful when many q-samples are needed from the same graph. Throughout the paper, query complexity is measured with respect to the sparse adjacency-list oracle $\mcl O_G$, where the neighbor list of each vertex is ordered lexicographically. Below is an informal statement of the theorem.

\begin{rtheorem}[\Cref{thm:w-optimal-query} (Informal)]
    Let $G=(V,E,w)$ be a connected weighted graph with $n$ vertices and $m\ge 100n$ edges, and let $\pi$ be its spanning tree distribution. 
    
    For any $\delta \in[0,1]$, there is an algorithm that prepares a q-sample of $\pi$ that admits a pre-processing step using $\toh\left(\sqrt{mn}+m^{1-\delta}\right)$ queries to the adjacency-list oracle $\mathcal O_G$, after which each q-sample can be generated using $\toh\left(n^{1+2\delta}\right)$ queries to $\mathcal O_G$.

    Alternatively, for any integer $k \geq 1$, there is an algorithm that prepares $k$ independent q-samples of $\pi$ with total query cost $\toh\left(\sqrt{kmn}\right)$, after the initial preprocessing of same cost.
\end{rtheorem}

We also prove a matching lower bound up to polylogarithmic factors, showing that the query complexity of our quantum algorithm is essentially optimal.

\begin{rtheorem}[\Cref{thm:lower-bound} (Informal)]
    Any algorithm preparing $k$ independent q-samples must perform at least~$\wtd{\Omega}(\sqrt{kmn})$ queries to $\orcl$.
\end{rtheorem}

In this work, we focus only on adjacency-list query complexity. A loose implementation-level accounting gives an $O(n^2)$ time bound (cf \Cref{Sec:time-complexity}) for both preprocessing and q-sampling, but we do not optimize runtime here. It is natural to expect that the time complexity of our algorithm can be brought closer to its query complexity with suitable quantum analogues of the dynamic tree data structures used in classical spanning-tree algorithms, such as link-cut trees~\cite{ALG+21c,ALV22c,RTF18j}. Developing such quantum dynamic-tree data structures is beyond the scope of this work.

\subsection{Applications}
  After the preprocessing phase, the spanning-tree q-sample can be used as a primitive for several downstream tasks. We describe two such applications: approximate counting of spanning trees and quantum-walk search for marked spanning trees.

\paragraph{Counting the number of spanning trees.}
For unweighted graphs, q-sampling can also be used to approximate the number of spanning trees. After the preprocessing step, we combine the q-sampling procedure along the slowly-varying cooling schedule with the state-of-the-art quantum partition-function approximation algorithm of Cornelissen and Hamoudi~\cite{CH23c}. This introduces an additional factor of $\toh(n^{1/4})$ over the q-sampling cost. Since the partition function is the weighted spanning-tree sum, and equals the number of spanning trees in the unweighted case, after preprocessing we obtain a $\toh(n^{5/4+2\delta})$-query algorithm for approximate spanning-tree counting, for any $\delta>0$. Including preprocessing, for $m\ge n^{5/4}$ we get a total query complexity of $\toh(\sqrt{m}\,n^{5/8})$.

For comparison, the recent algorithm of Liu, Peng, and Yang~\cite{LPY25p} estimates the spanning-tree count in time $\oh(m^{3/2})$. In sparse graphs, this improves over the previously best general bound $\toh(m+n^{15/8})$ by Chu, Gao, Peng, Sachdeva, Sawlani and Wang~\cite{CGP+20j}. Their result is a classical time bound, and since the algorithm reads the input graph, its query complexity is $\Theta(m)$. By contrast, our overall query complexity is $\toh(\sqrt{m}\,n^{5/8})$ for $m\ge n^{5/4}$, which is $o(m)$ whenever $m=\omega(n^{5/4})$. Thus, in this regime, our quantum counting algorithm is sublinear in the number of edges, including preprocessing.

\paragraph{Searching for marked spanning trees.}
The q-sample also gives the setup state needed in quantum-walk search algorithms. For instance, in the MNRS framework~\cite{MNRS11j}, for a reversible ergodic Markov chain with stationary distribution $\pi$, spectral gap $\gamma$, setup cost $S$, update cost $U$, checking cost $C$, and marked set of stationary measure at least $\alpha$, one can find a marked element with high probability using $\toh\left(S+\frac{1}{\sqrt{\alpha}}\left(\frac{1}{\sqrt{\gamma}}U+C\right)\right)$ queries.

In our setting, the setup cost corresponds to the cost of pre-processing and producing one q-sample, the update cost corresponds to the cost of reflection around the q-sample $\ket{\pi}$ which is summarized in \Cref{thm:w-optimal-query}. Thus our algorithm can be used as a black box inside quantum walk search algorithms for marked spanning trees, with an additional dependence on the marked fraction and the checking cost.

\subsection{Proof overview}
We provide high-level ideas behind the proof of \Cref{thm:w-optimal-query}. We begin with a simpler approach that illustrates the main ideas underlying our algorithm and highlights the key challenges that arise in achieving the improved bounds. We then describe how our techniques address these challenges to obtain the final result.

It is well known that a rapidly mixing reversible Markov chain with stationary distribution $\pi$ can be quantized to obtain a unitary operator $W$, referred to as the quantum walk operator, that enables efficient reflection about the q-sample $\ket{\pi}$. This construction is known as Szegedy quantization~\cite{Sze04c}. Given the ability to reflect about $\ket{\pi}$, one can start from an arbitrary state $\ket{\psi}$ and use Grover-style amplitude amplification to prepare $\ket{\pi}$. This requires $O(1/\braket{\psi}{\pi})$ applications of the quantum walk operator $W$. However, the overlap $\braket{\psi}{\pi}$ can be exponentially small for an arbitrary initial state $\ket{\psi}$. Thus, having a rapidly mixing Markov chain that allows for efficient sampling from $\pi$ does not, by itself, imply an efficient procedure for preparing the q-sample $\ket{\pi}$.

A standard approach for overcoming this issue is simulated annealing, which considers a sequence of slowly varying Markov chains whose stationary distributions interpolate between an easily preparable initial distribution and the target distribution. The parameters defining this sequence form a cooling schedule, which is chosen so that the q-samples of consecutive distributions have sufficiently large overlap. Consequently, their corresponding quantum walk operators can be used to transform one q-sample into the next, allowing the target q-sample to be prepared efficiently.

We design the following cooling schedule for spanning tree distributions: a sequence of slowly varying graphs that interpolates between an easily preparable state concentrated on a particular spanning tree and the target graph $G$. The schedule is defined using an admissible spanning tree $S$, namely, a spanning tree with non-negligible sampling probability, which can be identified using $\toh(\sqrt{mn})$ queries via the algorithm for finding the maximum weight-product spanning tree~\cite{DHHM06j}.

We parameterize the schedule by an inverse temperature $\beta \geq 0$. For each $\beta$, let $G_\beta$ be the graph obtained from $G$ by scaling the weight of every edge outside $S$ by $e^{-\beta}$ and leaving the edges of $S$ unchanged. Writing $\pi_\beta$ for the spanning tree distribution of $G_\beta$, we have, for every spanning tree $T$,
\[
    \pi_\beta(T)\propto e^{-\beta |T\setminus S|}\,w(T),
\]
where $|T\setminus S|$ is the number of edges of $T$ outside $S$. At $\beta=0$ we recover the target distribution, $\pi_0=\pi$. As $\beta$ grows, the distribution increasingly favors $S$, and as $\beta\to\infty$ it collapses onto $S$, so that $\pi_\infty=\mathds{1}_S$. Query access to every $G_\beta$ can be simulated using $\orcl$, the adjacency-list oracle for $G$.

We then show that a short schedule exists: there are inverse temperatures $\beta_l \geq \cdots \geq \beta_0=0$ such that (i) $\pi_{\beta_l}$ is concentrated mostly on the admissible tree $S$, (ii) consecutive distributions have constant overlap, $\braket{\pi_{\beta_i}}{\pi_{\beta_{i+1}}} = \Omega(1)$, and (iii) the schedule has length $l = \toh(\sqrt{n})$ (\Cref{prop:compute-adapt}), and it can be found using the adaptive-schedule procedure of Harrow and Wei~\cite{HW20c} using reflections around $\ket{\pi_{\beta}}$. Given an efficient way to reflect about each $\ket{\pi_\beta}$, we can compute such a schedule, start from $\ket{S}$ and walk down the schedule from $\beta_l$ to $\beta_0$, converting each q-sample into the next, to end with $\ket{\pi_{\beta_0}} = \ket{\pi}$.

One approach is to obtain these reflections by quantizing the standard up-down walk on spanning trees~\cite{CSV23p}. From a tree $T$, this walk adds an edge $f$ (the Up step), then removes an edge $e$ from the unique cycle in $T+f$ (the Down step), landing on $T+f-e$. The chain mixes in $\oh(m)$ steps~\cite{ALGV19c}, so the corresponding Szegedy walk operator $W_\beta$ reflects about $\ket{\pi_\beta}$ using $\toh(\sqrt{m})$ queries to $\orcl$. Using these, one can compute the schedule using the adaptive-schedule procedure takes $\toh(\sqrt{mn})$ queries to $\orcl$. Once the schedule is computed, traversing its $l=\toh(\sqrt{n})$ steps takes $\toh(\sqrt{n})\cdot\toh(\sqrt{m})=\toh(\sqrt{mn})$ queries, so the total cost of preparing the q-sample $\ket{\pi}$ is $\toh(\sqrt{mn})$.

This raises two natural questions. Can $k$ copies of $\ket{\pi}$ be prepared with fewer than the naive $\toh(k\sqrt{mn})$ queries? And can preprocessing reduce the cost of even a single q-sample? Our main result answers both in the affirmative.

\paragraph{Coherent Isotropization.}
Our approach is to quantize a different Markov chain from the standard up-down walk: one in which the transition probabilities of the Up and Down steps are biased by edge marginals. For the spanning tree distribution, the marginal of an edge $e$ is its leverage score $\ell(e)$:
\[
    \Pr_{T\sim \pi}(e\in T) = \ell(e)=w(e)\cdot\reff^G(e),
\]
where $\reff^G(e)$ is the effective resistance of $e$ in $G$. The Up step adds an edge sampled with probability proportional to its marginal, and the Down step removes an edge from the resulting cycle using a second distribution, also biased by marginals. We call this the marginal-aware walk, see \Cref{mc:marginal-aware} for a detailed description. Biasing by marginals acts as an implicit isotropic transformation: the walk behaves as if every edge had roughly the same marginal, that is, as if the distribution were in isotropic position. Its transition matrix has spectral gap $\wtd{\Omega}(1/n)$, so it mixes even faster than the standard up-down walk (\Cref{lem:margaw-walk-gap} and \cite{ALV22c}). Consequently, quantizing this chain allows us to implement a reflection about $\ket{\pi_\beta}$ using $\toh(\sqrt{n})$ applications of the resulting quantum walk operator. This walk is similar in spirit to the approach of~\cite{ALV22c}, but unlike theirs, ours still operates on the set of spanning trees, the same state space as the standard up-down walk. This preserves \emph{coherence}: the encoding of the q-sample is independent of $\beta$. Implementing the walk thus requires sampling edges according to the leverage score distribution. Since we quantize the chain for each $G_\beta$, we need effective-resistance overestimates for each of them, and we show how to avoid repeating this work across the schedule.

\paragraph{Preprocessing of effective resistances.}
We first compute a data structure that, once prepared, provides effective-resistance overestimates (whose sum over all edges is $O(n)$) for any $G_{\beta}$ without further queries to $\orcl$. A spectral sparsifier (a subgraph with fewer edges that approximates the Laplacian of $G$) suffices to compute such overestimates~\cite{SS11j}. The direct approach of building a new sparsifier for every $G_{\beta}$ is too expensive. Instead, we exploit the structure of $G_{\beta}$: for every $\beta$, its edge set is the same as that of $G$, and only the weights of the edges outside the admissible tree $S$ are rescaled, all by the same factor. We show that a sparsifier for the subgraph of $G$ outside $S$ is enough to obtain a sparsifier for any $G_{\beta}$ without further queries to $\orcl$. Such a sparsifier can be computed by the quantum sparsification procedure of Apers and de Wolf~\cite{AdW20c} using $\toh(\sqrt{mn})$ queries to $\orcl$. Thus, the same sparsification data, once computed, can be reused for every $G_{\beta}$ with only lightweight updates (\Cref{lem:spars-reuse}). This yields effective-resistance overestimates for every graph $G_{\beta}$ in the schedule using only $\toh(\sqrt{mn})$ queries to $\orcl$. This cost is paid once, during preprocessing, and is not incurred again for each $G_{\beta}$ or for each q-sample prepared afterwards.

\paragraph{Faster sampling from the leverage-score distribution.}
Once we have a data structure that gives effective-resistance overestimates for every $G_\beta$, we turn to implementing the quantum walk operators $W_\beta$ at each inverse temperature $\beta$ on the schedule, corresponding to the marginal-aware walk on $G_\beta$. A necessary primitive for implementing $W_\beta$ is a unitary $U$ that, given a spanning tree $\ket T$, prepares a superposition over the neighbors of $T$ in the chain, weighted by the square roots of the transition probabilities. Following the two steps of the walk, $U$ consists of an Up step that adds an edge $f$ sampled according to the leverage-score distribution of $G_\beta$, followed by a Down step that removes an edge $e$ from the unique cycle in $T+f$. The Down step acts only on the registers already holding $T$ and $f$ (with their weights): it finds the cycle by a search on $T+f$ and samples $e$ from a distribution on the cycle depending only on the weights stored in the registers, so it needs no queries to $\orcl$ (we discuss its time cost in \Cref{Sec:time-complexity}). The Up step, in contrast, does not depend on $T$ at all. It is the preparation of the state
\[
   \ket{\mu_\beta}=\frac{1}{\sqrt{Z}}\sum_{e\in E}\sqrt{\ell_\beta(e)}\ket e,
\]
where $\ell_\beta(e)=w_\beta(e)\,\reff^{G_\beta}(e)$ is the leverage score of $e$ in $G_\beta$, or equivalently its marginal under $\pi_\beta$, and $Z=n-1$ since leverage scores sum to $n-1$. Thus all queries made by $W_\beta$ are spent on q-sampling from the leverage-score distribution.

A natural way to prepare $\ket{\mu_\beta}$ is rejection sampling implemented with amplitude amplification, starting from a proposal state that is cheap to prepare. With the uniform proposal $\frac1{\sqrt m}\sum_e\ket e$, which costs $\toh(1)$ queries once the vertex degrees and their prefix sums are stored, and given effective-resistance overestimates of $G_\beta$, this uses $\toh(\sqrt{m/n})$ queries to $\orcl$. 

We improve on this with a \emph{$\rho$-heavy proposal}. For a threshold $\rho\in[n/m,1]$, let $\mcl E_\rho(\beta)$ be the set of edges of $G_\beta$ whose leverage score is at least $\rho$. The proposal gives each edge in $\mcl E_\rho(\beta)$ amplitude proportional to $\sqrt{\ell_\beta(e)}$ and every other edge the same amplitude $\sqrt\rho$:
\[
   \frac{1}{\sqrt{\widetilde{Z}}}\Big(
   \sum_{e\in \mcl{E}_{\rho}(\beta)} \sqrt{\ell_{\beta}(e)}\ket{e}
   + \sum_{e\notin \mcl{E}_{\rho}(\beta)} \sqrt{\rho}\ket{e}
   \Big).
\]
Since leverage scores sum to $n-1$, at most $(n-1)/\rho$ edges are heavy, so $\mcl E_\rho(\beta)$ can be stored explicitly. Given this stored set, the proposal can be prepared with $\toh(1)$ queries. Using amplitude amplification, this prepares $\ket{\mu_\beta}$ with $\toh(\sqrt{m\rho/n})$ queries to $\orcl$, improving on the uniform proposal by a factor $\sqrt\rho$. The restriction $\rho\ge n/m$ ensures the stated query bound.

The catch is that $G_\beta$ varies with $\beta$, so the leverage scores and hence the heavy edges may change along the schedule. We avoid recomputing the heavy edges at every $\beta$ by exploiting the monotonicity of leverage scores along the schedule. Recall that $S$ is the admissible spanning tree. As $\beta$ increases, leverage scores outside $S$ decrease, while those of edges in $S$ increase. Thus, the heavy edges at $\beta=0$, together with $S$, form an envelope $\mcl A_\rho$ containing all heavy edges throughout the schedule. In other words, every edge outside $\mcl A_\rho$ has leverage score below $\rho$ in every $G_\beta$. This envelope has size $O(n/\rho)$ and, together with the effective-resistance data structure, can be computed once during preprocessing with $\toh(\sqrt{mn/\rho})$ queries to $\orcl$. Storing these data therefore suffices to prepare the $\rho$-heavy proposal, and hence $\ket{\mu_\beta}$, for any $G_\beta$ using $\toh(\sqrt{m\rho/n})$ queries to $\orcl$, with no recomputation as $\beta$ changes. This is exactly the per-application query cost of $W_\beta$ in \Cref{prop:qw-operator}.

\paragraph{Preparing the q-sample and putting the costs together.}
Once these data structures are available, we start from $\ket{S}$ and anneal along the adaptive schedule in reverse order to prepare the target state $\ket{\pi}$. Since the adaptive schedule is constructed so that consecutive q-samples have sufficiently large overlap, we can efficiently transform $\ket{\pi_{\tau_j}}$ to $\ket{\pi_{\tau_{j-1}}}$ using the quantum walk operators at the relevant adaptive temperatures and fixed-point amplification~\cite{Gro05j} (\Cref{lem:q-sample-transform-walk}). Iterating these transformations prepares the target q-sample $\ket\pi$. The preprocessing phase builds the reusable data structures needed throughout the cooling schedule. Choosing the threshold parameter $\rho$ gives the tradeoff in \Cref{thm:w-optimal-query}: for any $\delta\in[0,1]$, the preprocessing cost is $\toh(\sqrt{mn}+m^{1-\delta})$, and each q-sample after preprocessing costs $\toh(n^{1+2\delta})$. If we want $k$ q-samples, we choose the threshold differently, giving preprocessing cost $\toh(\sqrt{kmn})$ and total q-sampling cost $\toh(\sqrt{kmn})$.

\section{Preliminaries}
We assume sparse adjacency-list access to $G$ through an oracle $\mcl O_G$ \cite{AdW20c}. Given a vertex $u$, the oracle returns $\deg(u)$; given additionally an index $i\in[\deg(u)]$, it returns the $i$th neighbor of $u$ and the corresponding edge weight. Neighbors are ordered lexicographically. This sparse-access model is standard in quantum algorithms for linear-system solving and Hamiltonian simulation \cite{HHL09j,CKS17j}. One has
\[
\mcl O_G\ket{u}\ket{0}
=
\ket{u}\ket{\deg(u)},
\]
and
\[
\mcl O_G\ket{u}\ket{i}\ket{0}\ket{0}
=
\ket{u}\ket{i}\ket{v}\ket{w(u,v)},
\]
where $v$ is the $i$th neighbor of $u$. Thus every undirected edge $\{u,v\}$ has two adjacency-list positions, $(u,i)$ and $(v,j)$, where $v$ is the $i$th neighbor of $u$ and $u$ is the $j$th neighbor of $v$. Throughout, adjacency-list oracle refers to this sparse-access oracle model.

Since adjacency lists are lexicographically sorted, we can also recover the adjacency-list index of an edge using binary search. Given $(u,v)$, we find the index $i$ such that $v$ is the $i$th neighbor of $u$ using $O(\log n)$ queries to $\mcl O_G$, implementing
\[
\ket{u}\ket{v}\ket{0}
\mapsto
\ket{u}\ket{v}\ket{i}.
\]

Our computational model is a quantum-accessible classical control model similar to \cite{AdW20c,AGJL25c}. For a graph $G$ with $n$ vertices, let $N$ denote the number of spanning trees of $G$. (Note that $N$ is generally exponential in $n$, and can be as large as $n^{n-2}$ for the complete graph.) The algorithm may run quantum subroutines on $\oh(\poly\log(N))$ qubits, make quantum queries to the input oracle $\mcl{O}_G$, and access classical data structures constructed during preprocessing through quantum-read/classical-write memory. A QRAM operation is either a classical write to such a data structure or a coherent quantum read, possibly in superposition.

Throughout, query complexity refers only to the number of queries to the input oracle $\mcl O_G$. QRAM operations and ordinary classical computation are not counted as input queries. When discussing time complexity, we count classical gates, quantum gates, input-oracle queries, and QRAM operations. The amount of QRAM used is the size of the data structures explicitly constructed during preprocessing, and will be specified when the relevant preprocessing procedure is invoked.

\subsection{Spanning Trees}

We consider only weighted connected graphs $G(V,E,w)$ with $n$ vertices, $m$ edges and non-negative edge weights $w$. Let $\Delta$ denote the set of spanning trees of $G$.
\begin{definition}[\sc Spanning tree distribution]
For $T\in\Delta$, define $w(T):=\prod_{e\in T}w(e)$. The spanning tree distribution $\pi_G$ is
\[
    \pi_G(T)=\frac{w(T)}{\sum_{S\in\Delta}w(S)}.
\]
When $G$ is clear from context, we simply write $\pi$.
\end{definition}

We recall the quantum speedup for classical sampling from $\pi$.
\begin{theorem}[\sc Quantum algorithm for classical sampling {\cite{AGJL25c}}]\label{thm:quantclassicalsample}
    There is a quantum algorithm that, given query access to $\mcl O_G$ for a connected weighted graph $G$, outputs a spanning tree drawn from a distribution $\epsilon$-close to $\pi$ in total variation distance with high probability. The algorithm uses $\toh(\sqrt{mn}\log(1/\epsilon))$ queries to $\mcl O_G$ and runs in $\toh(\sqrt{mn}\log(1/\epsilon))$ time.
\end{theorem}

We emphasize here that our goal is stronger: instead of having as an output a classical tree sampled from $\pi$, we prepare the coherent encoding of the distribution. In our algorithms, the tree register stores not only the edge identities but also their weights.
 We thus use the following canonical weighted encoding:

\begin{definition}[\sc Q-sample of $\pi$]
Let $\pi$ be the spanning tree distribution of $G=(V,E,w)$. We encode a spanning tree $T$ by listing its $n-1$ edges in lexicographic order, with each edge stored as an ordered pair of endpoints together with its weight in $G$. Thus the spanning tree writes
\[
    T
    =
    \bigl((u_1,v_1,w(u_1,v_1)),\ldots,(u_{n-1},v_{n-1},w(u_{n-1},v_{n-1}))\bigr),
\]
where $u_i<v_i$ for each $i$ and $(u_1,v_1)<_{\rm lex}\cdots<_{\rm lex}(u_{n-1},v_{n-1})$. The q-sample of $\pi$ is the coherent state
\[
    \ket{\pi}
    =
    \sum_{T\in\Delta}\sqrt{\pi(T)}\ket{T}.
\]
Here $\ket{T}$ denotes the canonical weighted edge-list encoding above.
\end{definition}

\begin{definition}[\sc Unique cut and cycle of a spanning tree]\label{def:cutcycl}
    For a spanning tree $T$ and an edge $e\in T$, let $\Cut(T-e)$ be the cut created by removing $e$. For an edge $f\notin T$, let $\Cycle(T+f)$ be the unique cycle created by adding $f$ to $T$.
\end{definition}

\begin{definition}[\sc Graph deletion, scaled addition and rescaling]\label{Def:del-scale}
Let $G=(V,E,w)$ and $K=(V,F,u)$ be weighted graphs on the same
vertex set. We define
\[
G\setminus K
:=
\bigl(V,E\setminus F,w|_{E\setminus F}\bigr),
\]
that is, $G\setminus K$ is obtained from $G$ by deleting the edges
in $E\cap F$.

For $c\geq0$, define $G+cK:=(V,E\cup F,w')$, where
\[
    w'(e)=
    \begin{cases}
    w(e), & e\in E\setminus F,\\
    c\cdot u(e), & e\in F\setminus E,\\
    w(e)+c\cdot u(e), & e\in E\cap F.
    \end{cases}
\]

For $c\geq0$ and $E'\subseteq E$, we define $G_{E',c}:=(V,E,w_{E',c})$, where
\[
    w_{E',c}(e)=
    \begin{cases}
    w(e), & e\notin E',\\
    c\cdot w(e), & e\in E'.
    \end{cases}
\]
\end{definition}

  \subsection{Leverage Scores and Spectral Sparsifiers}
  \begin{definition}[\sc Laplacian of a graph]
Let $G=(V,E,w)$ be an undirected weighted graph. For each edge $e=\{u,v\}\in E$, fix an arbitrary orientation and denote by $\chi_u$, the $u$-th standard basis vector in $\mathbb{R}^V$. The Laplacian of $G$ is 
\[
    L_G=\sum_{e = uv \in E} w(e)(\chi_u - \chi_v)(\chi_u - \chi_v)^\top .
\]
Equivalently, $L_G\in\mathbb{R}^{V\times V}$ is given entrywise by
\[
    \left(L_G\right)_{uv}= 
    \begin{cases}
        \sum_{e \in E: v \in e} w(e) & \text { if } u=v, \\ -w(uv) & \text { if }\left\{u,v\right\} \in E, \\ 0 & \text { otherwise. }
    \end{cases}
\]
The matrix $(\chi_u - \chi_v)(\chi_u - \chi_v)^\top$ is independent of the chosen orientation of $e$.
\end{definition}

\begin{fact}[\sc Laplacians are PSD]\label{fact:psd}
    The Laplacian matrix $L_G$ is positive semidefinite. 
\end{fact}

Graph sparsification produces a re-weighted graph with fewer edges, known as a graph spectral sparsifier. A graph spectral sparsifier of $G$ is a re-weighted subgraph that closely approximates the quadratic form of the Laplacian for any vector $x \in \mathbb{R}^{|V|}$. By a re-weighted subgraph of $G$, we mean a graph $\wtd{G}=(V,\wtd{E},\wtd{w})$ such that $\wtd{E} \subseteq E$, where $\wtd{w}$ assigns arbitrary nonnegative weights to edges in $\wtd{E}$ (not necessarily equal to the original weights $w$).

\begin{definition}[\sc Graph Spectral Sparsifier]\label{def:spec-sparsf}
    Let $G=(V, E, w)$ be a weighted graph, and let $\wtd{G}= (V, \wtd{E}, \wtd{w})$ be a re-weighted subgraph of $G$, where $\wtd{w}: \wtd{E} \rightarrow \mathbb{R}_{\geq 0}$ and $\wtd{E} \subseteq E$. 
    
    For any $\epsilon>0, \wtd{G}$ is an $\epsilon$-spectral sparsifier of $G$ if for any vector $x \in \mathbb{R}^V$, the following holds:
    $$
    \left|x^{\top} L_{\wtd{G}} x-x^{\top} L_G x\right| \leq  \epsilon \cdot x^{\top} L_G x .
    $$
\end{definition}

In the groundbreaking work \cite{SS11j}, the authors demonstrated that graphs can be efficiently sparsified by sampling $\toh(n)$ edges with probability proportional to the product of their weight and effective resistance (which is called the \emph{leverage score}). Let us now introduce the notion of effective resistance:

\begin{definition}[\sc Effective Resistance]\label{def:eff-res}
    Given a weighted graph $G=(V, E, w)$, the effective resistance $\reff(u,v)$ of a pair of vertices $u, v \in V$ is defined as
    $$
    \reff(u,v)=\left(\chi_u-\chi_v\right)^{\top} L_G^{+}\left(\chi_u-\chi_v\right)= \norm{\left(L_G^{+}\right)^{1 / 2}\left(\chi_u-\chi_v\right)}^2 .
    $$
    Here, $L_G^{+}$ is the Moore-Penrose inverse of the Laplacian matrix $L_G$.
\end{definition}

The effective resistances are connected to the marginals of the spanning tree distribution as follows.

\begin{lemma}[\sc Spanning tree marginals and leverage scores {\cite[Section~4.2]{LP17ba}}]
\label{lem:marg}
    Let $G=(V,E,w)$ be a weighted graph, and let $\pi$ be its weighted spanning tree distribution. For an edge $e=\{u,v\}\in E$, let $\reff(u,v)$ denote the effective resistance between its endpoints. Then $\Pr_{T\sim\pi}[e\in T] = w(e)\reff(u,v)$. The quantity $\ell_G(e):=w(e)\reff(u,v)$ is the leverage score of $e$. When the underlying graph is clear from context, we write $\ell(e)$ for $\ell_G(e)$.
\end{lemma}

Let us now recall a property of the marginals of the spanning tree distribution:

\begin{fact}[\sc Foster's theorem \cite{Fos49j}]
    \label{fact:foster}
    For any weighted graph $G=(V, E, w)$, the sum of the marginals of the spanning tree distribution $\pi$ is always at most $n-1$ i.e.,
    \[
        \sum\limits_{e\in E}\Pr_{T\sim \pi}[e\in T] = \sum\limits_{e\in E} \ell(e) \leq n-1
    \]
    with equality iff $G$ is connected.
\end{fact}

We usually do not have exact marginals and hence will use the following notion of approximate marginals and effective resistances in our algorithms.

\begin{definition}[\sc Marginal and effective-resistance overestimates]
    \label{def:marg-res-ovest}
    Consider a connected weighted graph $G=(V,E,w)$ with spanning tree distribution $\pi$. Let $\lambda \geq 1$ be a constant. A function $\wtd p:E\to\mathbb{R}_{\geq 0}$ is a $\lambda$-marginal overestimate if, for every $e\in E$, $\wtd p(e) \geq \Pr_{T \sim \pi}[e\in T]$ and $\sum_{e\in E}\wtd p(e)\le \lambda(n-1)$.

    Similarly, a function $\wtd R:V\times V\to\mathbb{R}_{\geq 0}$ is a
    $\lambda$-effective-resistance overestimate if, for every $u,v\in V$, it holds that $\reff(u,v) \leq \tdreff(u,v) \leq \lambda \cdot \reff(u,v)$. 
\end{definition}

\begin{remark}
    A $\lambda$-effective-resistance overestimate is a pointwise $\lambda$-approximation, whereas a $\lambda$-marginal overestimate only requires a pointwise lower bound and an upper bound on the total mass; it need not approximate each marginal within a factor $\lambda$.
\end{remark}

Apers and de Wolf~\cite{AdW20c} give a quantum speedup for graph sparsification over the classical algorithm of Spielman and Srivastava~\cite{SS11j}.

\begin{theorem}[\sc Quantum Speedup for Sparsification and Effective-resistances, adapted from {\cite{AdW20c}}]
\label{thm:spars-apersdw}
    There exists a quantum algorithm that, given query access to the adjacency-list oracle $\mcl O_G$ of a weighted graph $G=(V,E,w)$ with $n$ vertices and $m$ edges, and parameters $\epsilon \in [0,1/3]$ and $\eta>0$, outputs an explicit description of a $\epsilon$-spectral sparsifier of $G$ with $\toh(n/\epsilon^2)$ edges with probability at least $1-\eta$. The algorithm uses $\toh(\sqrt{mn}\ln(1/\eta)/\epsilon)$ queries to $\mcl O_G$ and runs in the same time.
\end{theorem}

Given a sparsifier, one can approximate the effective-resistances using the following theorem.

\begin{theorem}[\sc Effective-Resistance Overestimates from Sparsifiers, adapted from {\cite{AdW20c,SS11j}}]
\label{thm:eff-res-apersdw}
    Consider a weighted graph $G=(V,E,w)$ with $n$ vertices and $m$ edges. For any $\epsilon\in[0,1]$, suppose we are given an $\epsilon$-sparsifier $H$ of $G$. Then we can prepare query access to an efficiently implementable oracle $O^G_{\mcl R}$ such that, for all $u,v\in V$,
    \[
        O^G_{\mcl R}\ket{u}\ket{v}\ket{0}
        =
        \ket{u}\ket{v}\ket{\tdreff^G(u,v)},
    \]
    where $\tdreff(\cdot,\cdot)$ is a $(1+3\epsilon)$-effective-resistance overestimate of $G$. Preparing it requires $\toh(n/\epsilon^4)$ time and after that each query to $O^G_{\mcl R}$ runs in time $\wtd O(1/\epsilon^2)$ with no additional queries to $\mcl O_G$.
\end{theorem}

For our purposes, it suffices to apply the above theorem with constant $\epsilon$. We stress that, while the data structure $\mcl{O}^G_{\mcl{R}}$ allows computing all effective resistance overestimates~$\tdreff(u,v)$ without any additional query to $G$, computing the marginal overestimates $\wtd p(e) = w(e) \tdreff(u,v)$ still requires one query to $G$ per edge. This query cost is expected, as otherwise one could sometimes identify all edges of $G$ (those with nonzero marginals) using only~$\toh(\sqrt{mn})$ queries.

  \subsection{Markov Chains and Quantum Walks}
  \label{sec:quantum-walk}
  For probability distributions $q_1,q_2$ on a finite set $\mcl X$,  
the total variation distance between them is $d_{\rm TV}(p_1,p_2):=\frac12\sum_{x\in\mcl X}\abs{p_1(x)-p_2(x)}$. For quantum states, $\norm{\cdot}_2$ denotes the Euclidean norm.

\paragraph{Reversible Markov chains and spectral gaps.}
Let $P$ be the transition matrix of a reversible, irreducible, and aperiodic Markov chain on $\mcl X$, with stationary distribution $\pi$. We write $D_\pi:=\operatorname{diag}(\pi)$ and define the discriminant matrix of $P$ by $D(P):=D_\pi^{1/2}PD_\pi^{-1/2}$. The matrix $D(P)$ is symmetric and has the same eigenvalues as $P$. We write these eigenvalues as
$1=\lambda_0,\lambda_1,\ldots,\lambda_{|\mcl X|-1}$, ordered so that $\abs{\lambda_1}\ge\cdots\ge\abs{\lambda_{|\mcl X|-1}}$. The absolute spectral gap is
$\gamma_{\rm abs}(P):=1-\abs{\lambda_1}$. If $P$ has nonnegative spectrum, equivalently if $P$ is positive
semidefinite as an operator on $L^2(\pi)$, then this equals the usual spectral gap $\gamma(P):=1-\max_{j\ge1}\lambda_j$.

\paragraph{Labelled Szegedy walks.}
We use a labelled-transition formulation of Szegedy walks, in the convention used for quantum simulated annealing~\cite{WA08j}. Szegedy's original construction gives the spectral mapping theorem for reversible Markov chains~\cite{Sze04c}; see also~\cite{MNRS11j}. The standard equivalence between coined and Szegedy walks, including the use of a coin register and flip-flop shift, which maps a transition to the same transition viewed in reverse, is discussed in~\cite{Won17a,Won17j,PS17j}. 

Define $\mcl{K}:=\Span\{\ket{x}\ket{a}:x\in\mcl{X},\ a\in C_x \cup {0}\}$ as the Hilbert space over which the quantum walk operates. Here, $C_x$ is a finite set of transition labels for each $x\in\mcl X$ and $0$ is a distinguished blank transition label. In our application, $x$ is a spanning tree and the transition labels will be pairs of edges.

Each label $a\in C_x$ has endpoint $h_x(a)\in\mcl X$ and the probability of making a transition from $x$ along the label $a$ is $\Pr[a\mid x]$ . We require $P(x,y) = \sum_{{a\in C_x: h_x(a)=y}}\Pr[a\mid x]$. We assume that non-self-loop transitions are uniquely labelled: for every $x\neq y$ with $P(x,y)>0$, there is a unique label $a_{x,y}\in C_x$ with endpoint $y$. Self-loop labels have endpoint $x$ and total probability $P(x,x)$.

We also assume that the labelled transitions admit an involution $g$ such that, if $g(x,a)=(y,a')$, then $y=h_x(a)$ and $h_y(a')=x$. Self-loop labels are fixed by $g$.

Let $\mcl U$ be a transition-preparation unitary satisfying
\[
    \mcl{U}\ket{x}\ket{0}
    =
    \ket{x}\sum_{a\in C_x}\sqrt{\Pr[a\mid x]}\ket{a}.
\]
Let $\mcl S$ be the flip-flop shift, defined by $\mcl S\ket{x}\ket{a}:=\ket{g(x,a)} = \ket{h_x(a)}\ket{a'}$. Let $\mcl{A}:=\Span\{\ket{x}\ket{0}:x\in\mcl X\}$ and $\mcl{B}:=\mcl U^\dagger \mcl S \mcl U \mcl{A}$ be two subspaces with corresponding orthogonal projectors $\Pi_\mcl{A}$ and $\Pi_\mcl{B}$. Let $R_{\mcl{A}}:=2\Pi_{\mcl{A}}-I$ and
$R_{\mcl{B}}:=2\Pi_{\mcl{B}}-I$; the quantum walk operator is
\[
    W(P):=R_{\mcl{B}} R_{\mcl{A}} = \left(\mcl U^\dagger \mcl S \mcl U R_{\mcl{A}} \mcl U^\dagger \mcl S \mcl U\right) R_{\mcl{A}}.
\]

We call $\mcl{A}+\mcl{B}$ the busy subspace and denote by $\gamma_{\rm ph}(W)$ the smallest nonzero angular distance from $1$ among the eigenvalues of $W$ on this subspace. For a distribution $\pi$ on $\mcl X$, let $\ket{\pi} := \sum_{x\in\mcl X}\sqrt{\pi(x)}\,\ket{x}$ denote its q-sample.

\begin{theorem}[\sc Spectral mapping for labelled Szegedy walks]
\label{thm:labelled-szegedy}
Let $P$ be the transition matrix of a reversible, irreducible, and aperiodic Markov chain on $\mcl X$, with $\pi$ as the stationary~distribution. Then, on the busy subspace, the unique $1$-eigenvector of $W(P)$ is $\ket{\pi}\ket{0}$. Moreover,
\[
    \gamma_{\rm ph}(W(P))
    =
    \Omega(\sqrt{\gamma_{\rm abs}(P)}).
\]
\end{theorem}

The spectral mapping for coined and Szegedy walks is standard; see, e.g., \cite{Won17j,Won17a}. We restate and prove the result in our labelled-transition setting because our application represents neighbouring spanning trees by transition labels (pairs of edges), rather than by storing the entire neighbouring tree in the second register. This formulation avoids the costly coherent construction of a neighbouring spanning tree. The proof is deferred to \Cref{app:sec-coin-walk}.

Throughout the paper, whenever we refer to the quantum walk operator associated with a reversible Markov chain with transition matrix $P$, we mean the labelled Szegedy walk operator $W(P)$ constructed above, for the relevant labelled-transition realization of $P$. We abuse notation and write $\ket{\pi}$ for $\ket{\pi}\ket{0}$.

  \subsection{Up-Down walk over Spanning trees}
  We recall the up-down walk on spanning trees, which is reversible with respect
to the weighted spanning tree distribution and will underlie our sampling
procedure.

Given a current tree $T_i$, one step adds a uniformly random non-tree edge and
then removes an edge from the resulting cycle according to the target
distribution.

\begin{markovchain}[H]
\caption{Up-down walk from $T_0$}
\label{mc:ud}
\begin{algorithmic}[1]
\For{$i=0,1,2,\ldots$}
    \State Choose $f\notin T_i$ uniformly at random.
    \State Choose $e\in\Cycle(T_i+f)$ with probability $\Pr(e) \propto \pi(T_i+f-e)$.
    \State Set $T_{i+1}=T_i+f-e$.
\EndFor
\end{algorithmic}
\end{markovchain}

\begin{fact}[\sc Up-Down Walk mixes to $\pi$, \cite{ALGV19c}]
The up-down walk is reversible with respect to the weighted spanning tree distribution $\pi$, and hence $\pi$ is its stationary distribution. 
\end{fact}

We say that a spanning tree distribution is in isotropic position when all edge marginals are approximately equal. This discrete notion of isotropic position was introduced by Anari and Dereziński~\cite{AD20c}, in analogy with isotropic position for continuous log-concave distributions.

\begin{definition}[\sc Isotropic Distributions]
    \label{def:iso-dist}
    Suppose $G(V,E,w)$ is a connected, weighted graph with the spanning tree distribution $\pi$. We say that the distribution $\pi$ is isotropic if all the marginals of edges are equal to each other i.e.,
    \[
        \Pr_{T\sim \pi}[e \in T] = \frac{(n-1)}{m}\text{ for all } e \in E. 
    \]
    Moreover, we call $\pi$ to be nearly-isotropic distribution if for all $e \in E$, its marginal satisfies the following,
    \[
        \Pr_{T\sim \pi}[e \in T] \leq \frac{c(n-1)}{m} \text{ for some constant } c >1.
    \]
\end{definition}

The up-down walk has strong mixing guarantees: \cite{ALGV19c} proved a
general spectral-gap lower bound, and \cite{ALV22c} gave an improved
log-Sobolev bound when the spanning tree distribution is nearly isotropic. We use the standard comparison between the log-Sobolev constant $\rho$ and the spectral gap $\gamma$ for reversible Markov chains, namely $\rho\le 2\gamma$ \cite{CSV23p}, so log-Sobolev lower bounds imply spectral-gap lower bounds up to constants.

\begin{theorem}[\sc Mixing guarantees for the up-down walk, adapted from \cite{ALGV19c,ALV22c}]
\label{thm:up-down-gap}
Let $G=(V,E,w)$ be a connected weighted graph with $n$ vertices, $m$ edges,
and non-negative edge weights. Let $\pi$ be the spanning tree distribution of
$G$, and let $\gamma_{\mathrm{ud}}$ denote the spectral gap of the up-down walk on the set $\Delta$ of spanning trees of $G$. Then the up-down walk satisfies $\gamma_{\mathrm{ud}} \ge \Omega\!\left({1}/{m}\right)$.
Moreover, if $m\ge 100n$ and $\pi$ is nearly-isotropic, then the up-down walk satisfies the improved bound $\gamma_{\mathrm{ud}} \geq \Omega\!\left({1}/{n\log^2 n}\right)$.
\end{theorem}

Anari and Dereziński~\cite{AD20c} introduced a subdivision procedure that
uses marginal overestimates to transform a graph into a multigraph whose
spanning tree distribution is nearly isotropic. The procedure replaces each
edge by a number of parallel copies proportional to its marginal overestimate
and splits the original weight evenly among those copies.

\begin{definition}[\sc Isotropic multiplicity]
\label{def:isotropic-multiplicity}
    Let $G=(V,E,w)$ be a weighted graph with spanning-tree distribution $\pi_G$, and let $\wtd p_G$ be a $\lambda$-marginal overestimate for $\pi_G$. For each $e\in E$, define the isotropic multiplicity of $e$ induced by $\wtd p_G$ as
    \[
        t_G(e):=
        \left\lceil
        \frac{m}{\lambda(n-1)}\wtd p_G(e)
        \right\rceil.
    \]
\end{definition}

\begin{proposition}[\sc Isotropic Transformation;
{\cite{ALV22c}}]
\label{prop:iso-transform}
    Let $G=(V,E,w)$ be a weighted graph with $n$ vertices and $m$ edges, and let $\pi$ be its spanning tree distribution. Let $\wtd{p}_G:E\to\mathbb{R}_{\ge 0}$ be a $\lambda$-marginal overestimate for $\pi$ (\Cref{def:marg-res-ovest}). The isotropic transform of $G$ with respect to $\wtd{p}_G$ is the multigraph $G'=(V,E',w')$ obtained by replacing each edge $e\in E$ by $t_G(e)$ parallel copies $e^{(1)},\ldots,e^{(t_G(e))}$, each with weight $w(e)/t_G(e)$. Let $\pi'$ denote the
    spanning-tree distribution of $G'$.

    If $T'=\{e_1^{(j_1)},\ldots,e_{n-1}^{(j_{n-1})}\}$ is a spanning tree in $G'$ whose edges are copies of a spanning tree $T=\{e_1,\ldots,e_{n-1}\}$ in $G$, then $\pi'(T') = {\pi(T)}/{\prod_{e\in T}t_G(e)}$. Equivalently, sampling from $\pi'$ can be done by first sampling $T\sim\pi$ and then choosing, independently for each $e\in T$, one of its $t_G(e)$ copies uniformly at random.

    Moreover, the transformed distribution is nearly isotropic (\Cref{def:iso-dist}) and the number of edges in $G'$ satisfies $|E'|\le 2m$.
\end{proposition}

The transformed multigraph viewpoint is useful for proving near-isotropy, but
one need not explicitly construct $G'$ in order to implement the corresponding
dynamics. Instead, we will present a way in \Cref{Sec:coherentIso} to work directly on the original graph and bias the choice of the added edge according to the number of copies it has in $G'$. 

  \subsection{Simulated Annealing}
  \label{sec:sim-anl}
  We recall the standard simulated-annealing framework involving classical Hamiltonians. Let $\Delta$ be a finite state space, let $\mu:\Delta\to\mathbb{R}_{\geq 0}$ be a base weight function, and let $H:\Delta\to \{0,1,\dots,n\}$ be a classical Hamiltonian of degree $n$. We define the partition function at inverse temperature $\beta \geq 0$ as
\[
    Z(\beta)=\sum_{x\in\Delta}e^{-\beta H(x)}\cdot \mu(x).
\] 
Such a partition function corresponds to the normalization factor of the Gibbs distribution at inverse temperature $\beta$ given by $\pi_\beta(x)=\frac{1}{Z(\beta)}e^{-\beta H(x)}\cdot \mu(x)$. For each inverse temperature $\beta$, we write $\ket{\pi_{\beta}}:=\sum_{x\in\Delta}\sqrt{\pi_{\beta}(x)}\ket{x}$ for the q-sample associated with the Gibbs distribution $\pi_{\beta}$.

A cooling schedule is a sequence of inverse temperatures $\beta_0< \cdots<\beta_l$. We study mainly the following type of cooling schedules.

\begin{definition}[\sc $B$-slowly varying schedule]
    For any partition function $Z(\beta)$ which is log-convex and $B > 1$, a cooling schedule $\beta_0< \dots < \beta_l$ is called $B$-slowly varying if 
    \[
        \abs{\braket{\pi_{\beta_i}}{\pi_{\beta_{i+1}}}}^2 \geq 1/B.
    \]
    This is a \emph{warm-start} property between $\pi_{\beta_i}$ and $\pi_{\beta_{i+1}}$ and ensures, we can \emph{anneal} i.e., map, $\ket{\pi_{\beta_i}}$ to $\ket{\pi_{\beta_{i+1}}}$ with less effort.
\end{definition}

When the partition function is log-convex, the existence of short slowly varying schedules was first shown by Štefankovič, Vempala, and Vigoda~\cite{SVV09j}. We use the following form stated by Harrow and Wei~\cite{HW20c}. However, it should be noted that the partition functions considered in the above works do not include the additional weight pre-factor $\mu$. Nevertheless, the same results extend to this setting with only minor modifications, since $\mu(x)$ is independent of $\beta$ and depends only on the state $x$. 

\begin{theorem}[\sc Quantum algorithm for computing a slowly-varying schedule, adapted from~\cite{HW20c}]
\label{thm:q-adapt-anneal}
Let $\{\pi_\beta\}_{\beta\in[0,\beta^{\star}]}$ be a family of Gibbs distributions with log-convex partition function $Z(\beta)$. There exists an $e^2$-slowly varying schedule $0=\beta_0<\cdots<\beta_l=\beta^{\star}$ of length $l \leq \toh\left(\sqrt{\ln{\frac {Z(0)}{Z(\beta^{\star})}}}\right)$.

Moreover, if a state $\ket{\wtd{\pi}_0}$ satisfying $\norm{\ket{\wtd{\pi}_0}-\ket{\pi}}_2\leq\epsilon$ can be prepared and, for every $\beta\in[0,\beta^\star]$, there is a reversible Markov chain $\mcl{M}_\beta$ with stationary distribution $\pi_\beta$ and absolute spectral gap at least $\delta$, together with an implementable quantum walk operator $W_\beta$, then for any $\eta>0$, a quantum algorithm constructs such a schedule with probability at least $1-\eta$ using
\[
\toh\left(\frac{l}{\sqrt{\delta}}\ln\left(\frac{1}{\eta}\right)\ln\left(\frac{1}{\epsilon}\right)\right)
\]
applications of $W_\beta$ and $W_\beta^{-1}$ at adaptively chosen values of $\beta$.
\end{theorem}

\section{Main Result}
In this section, we prove our main result, \Cref{thm:w-optimal-query}. For any $\rho\in[n/m,1]$, after a preprocessing step using $\toh(\sqrt{mn/\rho})$ queries to the adjacency-list oracle, our algorithm prepares the q-sample $\ket{\pi}$ of the spanning tree distribution $\pi$ using $\toh(\sqrt{mn\rho})$ queries.

\begin{theorem}[\sc Quantum Algorithm for Q-Sample of spanning tree distribution]
\label{thm:w-optimal-query}
Let $G=(V,E,w)$ be a connected weighted graph with $n$ vertices and $m\ge 100n$ edges, and let $\pi$ be its spanning tree distribution. Let $\rho\in[n/m,1]$. For any
$\epsilon,\eta>0$, there is a quantum algorithm which, with probability at
least $1-\eta$, outputs a state $\ket{\widetilde\pi}$ satisfying $\norm{\ket{\widetilde\pi}-\ket{\pi}}_2\le \epsilon$. The algorithm has the
following query costs with respect to the adjacency-list oracle $\mcl O_G$:
\begin{enumerate}
    \item preprocessing cost:
    $\toh\left(\sqrt{{mn}/{\rho}}\cdot \ln(1/\eta)\right)$ queries;
    \item reflection cost around $\ket{\pi}$:
    $\toh\left(\sqrt{m\rho}\right)$ queries;
    \item q-sampling cost:
    $\toh\left(\sqrt{mn\rho}\cdot \ln(1/\epsilon)\right)$ queries.
\end{enumerate}

In particular, preparing one q-sample uses
$\toh\left(\sqrt{mn\rho}\cdot \ln(1/\epsilon)\right)$ queries to $\mcl O_G$
after the successful completion of the preprocessing.
\end{theorem}

\Cref{thm:w-optimal-query} exposes a tradeoff, controlled by $\rho$, between the preprocessing cost and the cost of each subsequent q-sample. We instantiate it in two ways, which give the two statements in the informal version of the theorem in the introduction. Throughout, we suppress the dependence on $\epsilon$ and $\eta$ inside $\toh(\cdot)$.

\paragraph{Preparing $k$ q-samples.}
Choosing $\rho = 1/k$ for $k\in[1,m/n]$, we obtain that the preprocessing cost is $\toh(\sqrt{kmn})$ and the cost of q-sampling is $\toh(\sqrt{mn/k})$. Thus for $k$ independent q-samples, we obtain a total query cost of $\toh(\sqrt{kmn}) + k\cdot\toh(\sqrt{mn/k}) = \toh(\sqrt{kmn})$.

\paragraph{Cheaper q-samples after preprocessing.}
Alternatively, for $\delta\in[0,1]$, we choose the parameter $\rho=\min\{1,\,n\,m^{2\delta-1}\}$, which lies in $[n/m,1]$ since $m\ge 100n$ to obtain the following. If $n\,m^{2\delta-1}\le 1$, the preprocessing cost is $\toh(\sqrt{mn/\rho})=\toh(m^{1-\delta})$ and each q-sample costs $\toh(\sqrt{mn\rho})=\toh(n\,m^{\delta})$. Otherwise, $\rho=1$ and both costs are $\toh(\sqrt{mn})$, which is at most $\toh(n\,m^{\delta})$. In either case, the preprocessing cost is $\toh(\sqrt{mn}+m^{1-\delta})$, and since $m\le n^2$, each q-sample costs $\toh(n\,m^{\delta})\le\toh(n^{1+2\delta})$.

Before proving the theorem, we formalize the cooling schedule described in the introduction.

\begin{definition}[\sc Admissible Spanning Tree, \cite{AD20c}]
    A spanning tree $S \in \Delta$ is \emph{admissible} if $\pi(S) \geq {1}/{(2n^{n-2})}$.
\end{definition}

\begin{definition}[\sc $\beta$-scaled weight function and spanning tree distribution]\label{def:beta-scale}
    For every $\beta\in\mathbb{R}$, define the $\beta$-scaled weight function
    $w^S_\beta:E\rightarrow \mathbb{R}_{\geq 0}$ by
    \[
    w^S_\beta(e)=
    \begin{cases}
    e^{-\beta}\cdot w(e), & \text{if } e\notin S,\\
    w(e), & \text{if } e\in S.
    \end{cases}
    \]
    Let $G^S_\beta=(V,E,w^S_\beta)$ be the weighted graph induced by $w^S_\beta$, and let
    $\pi_\beta$ denote the spanning tree distribution of $G^S_\beta$. Note that $G^S_{\beta}$ is same as the rescaled graph $G_{E\setminus S,e^{-\beta}}$. In particular,
    for any spanning tree $T$,
    \[
    \pi_\beta(T)\propto \prod_{e\in T}w^S_\beta(e)
    = e^{-\beta |T\setminus S|}\cdot
    \prod_{e\in T}w(e) = e^{-\beta |T\setminus S|}\cdot
   w(T).
    \]
\end{definition}

At $\beta=0$, $\pi^S_0$ is the original spanning tree distribution, i.e., $\pi_0^S=\pi$. As $\beta\to\infty$, only the tree $S$ survives, so $\pi_\infty^S=\mathds{1}_S$. When $S$ is fixed by context, we omit the superscript $S$ and write $w_\beta$, $G_\beta$, and $\pi_\beta$ for $w_\beta^S$, $G_\beta^S$, and $\pi_\beta^S$, respectively. 

We now describe the four intermediate results used in the proof of \Cref{thm:w-optimal-query}. The first one gives reusable effective-resistance overestimates for any sequence of graphs obtained by rescaling the same set of edges.

\begin{rproposition}[\Cref{prop:er-ovest} ({\sc Reusable Effective-resistance data structure})]
    Let $G=(V,E,w)$ be a weighted graph with $n$ vertices and $m$ edges, and let $F\subseteq E$ be stored explicitly together with its edge
    weights. For any $\eta > 0$, there is a quantum algorithm that, with probability at least $1-\eta$, uses $\toh(\sqrt{mn}\cdot\ln(1/\eta))$ queries to the adjacency-list oracle $\orcl$ to construct a data structure $\mathcal H$ for the family of rescaled graphs $\{G_{E\setminus F,c}: c\geq0\}$.

    For any $c\geq0$, $\mathcal H$ can be used to construct an oracle $O_{\mathcal R}^{G_{E\setminus F,c}}$ giving access to $2$-effective-resistance overestimates in $G_{E\setminus F,c}$. Constructing this oracle, as well as querying it, requires no further queries to $\orcl$.
\end{rproposition}

The data structure $\mcl H$ is the only place where the preprocessing looks at the whole graph, and it is built once for the entire family $\{G_{E\setminus F,c}\}_{c\ge 0}$. This is what makes a long cooling schedule affordable. All graphs on the schedule differ from $G$ only by rescaling the edges outside a fixed spanning tree, so their effective resistances can all be overestimated from the same sparsification data. The length of the schedule therefore affects only the time complexity, not the query complexity.

To implement the quantum walk operator $W_\beta$ corresponding to the marginal-aware walk, we need to implement its Up step. This step samples an edge $e$ with probability proportional to the isotropic multiplicity $t_G(e)$, and hence requires preparation of the corresponding quantum state:

\begin{definition}[\sc Marginal state]
\label{def:marginal-state}
    Let $G=(V,E,w)$ and fix a $\lambda$-marginal overestimate $\wtd p_G$, for some constant $\lambda\geq1$, inducing isotropic multiplicities $t_G(e)$ as in \Cref{def:isotropic-multiplicity}. The corresponding marginal state is
    \[
        \ket{\mu_G}
        :=
        \frac{1}{\sqrt{\sum_{e\in E}t_G(e)}}
        \sum_{e\in E}\sqrt{t_G(e)}\,\ket e.
    \]
\end{definition}

\begin{rproposition}[\Cref{prop:computable-heavy-envelope} ({\sc Reusable marginal data structure})]
    Consider any weighted graph $G=(V,E,w)$ with $n$ vertices and $m$ edges, and let $F\subseteq E$ be stored explicitly together with its edge weights. For any $\rho\in[n/m,1]$ and $\eta>0$, there is a quantum algorithm that constructs a data structure $\margds$ providing coherent query access to a set of edges, together with their weights, such that every edge $e\in E$ outside this set has marginal less than $\rho$ in every graph in the family $\{G_{E\setminus F,c}:0<c\leq1\}$. The algorithm succeeds with probability at least $1-\eta$ and uses $\toh(\sqrt{mn/\rho}\cdot\ln(1/\eta))$ queries to the adjacency-list oracle $\orcl$.
    
    Moreover, after constructing $\margds$, for any $0<c\leq1$ it can be used to prepare the marginal state $\ket{\mu_{G_{E\setminus F,c}}}$ corresponding to the $2$-marginal overestimate
    $\wtd p_{G_{E\setminus F,c}}(e):=w_{G_{E\setminus F,c}}(e)\tdreff^{G_{E\setminus F,c}}(e)$
    (\Cref{def:marginal-state}), using $\toh(\sqrt{m\rho/n})$ queries to $\orcl$, given coherent query access to $2$-effective-resistance overestimates $\tdreff^{G_{E\setminus F,c}}$.
\end{rproposition}

We now have, for every temperature on the schedule, coherent access to effective-resistance overestimates and a way to prepare the marginal state $\ket{\mu_G}$ at cost $\toh(\sqrt{m\rho/n})$. Crucially, the stored set of heavy edges is \emph{static}. It is computed once and is valid for every graph in the family $\{G_{E\setminus F,c}:0<c\le 1\}$, so it does not need to be recomputed as $\beta$ changes. These two ingredients are exactly what is needed to implement a marginal-aware up-down walk on spanning trees directly on the original graph, without ever explicitly rescaling or isotropizing it. The next proposition records the resulting quantum walk operator.

\begin{rproposition}[\Cref{prop:qw-operator} ({\sc Quantum walk operator with large gap})]
    Let $G=(V,E,w)$ be any weighted graph with $n$ vertices and $m\ge 100n$ edges, and let $\Delta$ denote the set of spanning trees of $G$. There exists a quantum walk operator $W$, associated with a reversible Markov chain on $\Delta$, whose unique $1$-eigenvector is $\ket{\pi}$, where $\pi$ is the spanning-tree distribution of $G$, and whose phase gap is at least $\wtd{\Omega}(1/\sqrt n)$.

    Moreover, for any $\rho\in[n/m,1]$, one application of $W$ and its inverse can be implemented using $\toh(\sqrt{m\rho/n})$ queries to $\mcl O_G$, given:
    \begin{itemize}
        \item coherent query access to $2$-effective-resistance overestimates for all edges of $G$;
        \item the degree of every vertex of $G$; and
        \item a stored set of edges such that every edge outside this set has marginal under $\pi$ less than $\rho$, with each stored edge given together with its weight and its two adjacency-list positions.
    \end{itemize}
\end{rproposition}

The quantum walk operators serve both to construct the cooling schedule and to move a q-sample between consecutive temperatures. Their $\widetilde{\Omega}(1/\sqrt n)$ phase gap, quadratically larger than the $\widetilde{\Omega}(1/n)$ spectral gap of the underlying chain, reduces each transition to $\widetilde{O}(\sqrt n)$ walk steps. It remains to bound the number of transitions: following Harrow and Wei~\cite{HW20c}, the schedule can be constructed from the operators $W_\beta$ themselves and has length $\widetilde{O}(\sqrt n)$. Thus its cost is measured in walk-operator applications, each requiring $\widetilde{O}(\sqrt{m\rho/n})$ queries.

\begin{rproposition}[\Cref{prop:compute-adapt} ({\sc Computing a slowly-varying schedule})]
    Let $G=(V,E,w)$ be a weighted graph with $n$ vertices and $m$ edges. Fix an admissible spanning tree $S$, let $\beta^{\star}:=\ln(16n^{n-2}/\epsilon^2)$ for $\epsilon>0$, and let $\{G_\beta\}_{\beta\in[0,\beta^{\star}]}$ be the $\beta$-scaled graphs from \Cref{def:beta-scale}, with spanning tree distributions $\{\pi_\beta\}$. 

    Suppose the admissible spanning trees $S$ is stored explicitly together with its edge weights, then there is a quantum algorithm that, for any $\eta>0$, computes an $e^2$-slowly varying schedule $0=\beta_0<\cdots<\beta_l=\beta^{\star}$ of length $l \leq \toh(\sqrt n)$ with $\beta_0=0$, $\beta_l=\beta^{\star}$ with probability at least $1-\eta$. The algorithm uses $\toh(n\cdot \ln(1/\eta)\ln(1/\epsilon))$ applications of the operators $W_\beta$ obtained by applying \Cref{prop:qw-operator} to $G_{\beta}$.
\end{rproposition}

Combining the four ingredients gives \Cref{alg:q-sample}. The preprocessing phase computes an admissible tree and all reusable data, and then finds the schedule. The q-sampling phase starts from the point mass on the tree and anneals down the schedule to $\beta_0=0$.

\begin{algorithm}[ht!]
\caption{\sc Quantum Sampling of Random Spanning Trees}
\label{alg:q-sample}
    \begin{algorithmic}[1]
    \Require{$G=(V,E,w)$ with $n=|V|$ and $m=|E|$, parameter
    $\rho\in[n/m,1]$, accuracy $\epsilon$, and failure probability $\eta$.}
    \Ensure{A state $\ket{\widetilde\pi}$ satisfying
    $\norm{\ket{\widetilde\pi}-\ket{\pi}}_2\le\epsilon$.}

    \Statex
    \Statex \textsc{Preprocessing}
    \State Compute $S$, a maximum weight-product spanning tree of $G$
    (\Cref{prop:sample-admissible}, \cite{AGJL25c}).

    \State Compute $\mcl H$, a reusable effective-resistance data structure for
    \Statex \hspace{2pt}
    the family $\{G_{E\setminus S,e^{-\beta}}:\beta\geq 0\}$
    (\Cref{prop:er-ovest}, \cite{AdW20c}).

    \State Compute $\mcl L_{\rho}$, the reusable marginal data structure for the family $\{G_{E\setminus S,e^{-\beta}}:\beta \geq0\}$
    (\Cref{prop:computable-heavy-envelope}).

    \State Set $\beta^{\star}\leftarrow \ln(16n^{n-2}/\epsilon^2)$.

    \State Compute an $e^2$-slowly-varying schedule
    $0=\beta_0<\cdots<\beta_l=\beta^{\star}$ of length $l=\toh(\sqrt n)$
    \Statex \hspace{2pt}
    (\Cref{prop:compute-adapt}, \cite{HW20c}).

    \For{$j=l,l-1,\ldots,0$}
        \State Construct from $\mcl H$ an oracle $O_{\mcl R}^{G_{\beta_j}}$
        giving query access to $2$-effective-resistance
        \Statex \hspace{2pt+\algorithmicindent}
        overestimates in
        $G_{\beta_j}=G_{E\setminus S,e^{-\beta_j}}$. (\Cref{prop:er-ovest,thm:eff-res-apersdw})
    \EndFor

    \Statex
    \Statex {\sc Quantum sampling}

    \State Prepare $\ket{S}$; by \Cref{lem:pointmass},
    $\norm{\ket{S}-\ket{\pi_{\beta_l}}}_2\le\epsilon/2$.

    \For{$j=l,l-1,\ldots,1$}
        \State Apply $U_j$, which transforms
        $\ket{\pi_{\beta_j}}$ to $\ket{\pi_{\beta_{j-1}}}$
        using \Cref{lem:q-sample-transform-walk} and \Cref{prop:qw-operator}.
    \EndFor

    \State Output the resulting state $\ket{\widetilde\pi}$.

    \end{algorithmic}
\end{algorithm}

We are now ready to prove \Cref{thm:w-optimal-query} by combining the preceding propositions.

\begin{proof}[Proof of \Cref{thm:w-optimal-query}]
We first analyze the preprocessing phase and then the q-sampling phase.

\paragraph{Preprocessing: correctness and failure probability.}
There are four randomized steps: computing the admissible tree, constructing $\mcl H$, constructing $\margds$, and computing the slowly-varying schedule. We allocate failure probability $\eta/4$ to each. By a union bound, all four succeed with probability at least $1-\eta$. Below, we condition on the success of the earlier steps when analyzing a later one.

\paragraph{Preprocessing: cost.}
\emph{Admissible tree.} By \Cref{prop:sample-admissible}, with probability at least $1-\eta/4$ we obtain an admissible spanning tree $S$, given as a list of edges and their weights, using $\toh(\sqrt{mn}\ln(1/\eta))$ queries to $\mcl O_G$. We convert $S$ to the canonical weighted edge-list encoding by replacing each returned pair $(u,i)$ with $(\min\{u,v\},\max\{u,v\},w(u,v))$, where $v$ is the $i$th neighbor of $u$, and sorting the $n-1$ edges lexicographically. This uses $O(n)$ queries. We also read the degree of every vertex with $n$ queries. Since rescaling weights does not change the adjacency structure, this degree information is valid for every $G_\beta$. Both costs are $O(n)\le O(\sqrt{mn})$ because $m\ge n$.

\emph{Effective resistances.} Apply \Cref{prop:er-ovest} with $F=S$ and failure probability $\eta/4$. This yields $\mcl H$ for the family $\{G_{E\setminus S,c}:c\ge 0\}$ using $\toh(\sqrt{mn}\ln(1/\eta))$ queries. For every $\beta\ge0$, the oracle $O_{\mcl R}^{G_\beta}$ for $G_\beta=G_{E\setminus S,e^{-\beta}}$ can then be built and evaluated with no further queries to $\mcl O_G$.

\emph{Heavy edges.} Apply \Cref{prop:computable-heavy-envelope} with $F=S$, threshold $\rho$, and failure probability $\eta/4$. This yields $\margds$ using $\toh(\sqrt{mn/\rho}\ln(1/\eta))$ queries. It provides a stored set $\mcl{E}_{\rho}$ of edges with their weights such that every edge outside $\mcl{E}_{\rho}$ has marginal less than $\rho$ in every $G_{E\setminus S,c}$ with $0<c\le1$. In particular, this holds for every $G_\beta$ with $\beta\ge0$, since $c=e^{-\beta}\in(0,1]$. 

\emph{Slowly-varying schedule.} Set $\beta^\star=\ln(32n^{n-2}/\epsilon^2)$. Apply \Cref{prop:compute-adapt} with failure probability $\eta/4$; $S$ is stored explicitly with its weights, as required. It returns an $e^2$-slowly-varying schedule $0=\beta_0<\cdots<\beta_l=\beta^\star$ of length $l\le\toh(\sqrt n)$, using $\toh(n\ln(1/\eta)\ln(1/\epsilon))$ applications of the operators $W_\beta$ from \Cref{prop:qw-operator}, at points $\beta$ chosen adaptively. We check that each application is available at every such $\beta$. The three hypotheses of \Cref{prop:qw-operator} for $G_\beta$ are satisfied: (i) $\mcl H$ yields coherent access to $2$-effective-resistance overestimates for all edges of $G_\beta$; (ii) the degree table is stored; (iii) $\mcl{E}_{\rho}$ is stored with their weights such that every edge outside this set has marginal less than $\rho$. Hence $W_\beta$ and $W_\beta^{-1}$ each cost $\toh(\sqrt{m\rho/n})$ queries, and the schedule computation costs
\[
    \toh\big(n\ln(1/\eta)\ln(1/\epsilon)\big)\cdot\toh\big(\sqrt{m\rho/n}\big)
    =\toh\big(\sqrt{mn\rho}\cdot\ln(1/\eta)\ln(1/\epsilon)\big).
\]
Finally, constructing and storing $O_{\mcl R}^{G_{\beta_j}}$ for each $j$ requires no further queries.

\emph{Total.} Summing the four contributions, the preprocessing succeeds with probability at least $1-\eta$ and uses
\[
    \toh\Big(\big(\sqrt{mn}+\sqrt{mn/\rho}+\sqrt{mn\rho}\big)\ln(1/\eta)\ln(1/\epsilon)\Big)
    =\toh\Big(\sqrt{mn/\rho}\cdot\ln(1/\eta)\ln(1/\epsilon)\Big)
\]
queries, since $\rho\le1$. The dependence on $\epsilon$ enters only through $\beta^\star$ and is suppressed in $\toh(\cdot)$, as in the statement.

\paragraph{Q-sampling: correctness.}
Condition on successful preprocessing. Since $S$ is already stored in canonical form, the basis state $\ket{S}$ is prepared without further queries. By \Cref{lem:pointmass} and the choice of $\beta^\star$, we have $\norm{\ket{S}-\ket{\pi_{\beta_l}}}_2\le\epsilon/2$.

For $j=l,\ldots,1$, let $U_j$ be the unitary from \Cref{lem:q-sample-transform-walk}, implemented with the walk operators $W_{\beta_j}$ and $W_{\beta_{j-1}}$ of \Cref{prop:qw-operator}. It satisfies $\norm{U_j\ket{\pi_{\beta_j}}-\ket{\pi_{\beta_{j-1}}}}_2\le \epsilon/(2l)$.
Let $U:=U_1U_2\cdots U_l$, so that $U_l$ is applied first. Telescoping gives
\[
    U\ket{\pi_{\beta_l}}-\ket{\pi_{\beta_0}}
    =\sum_{j=1}^{l}U_1\cdots U_{j-1}\Big(U_j\ket{\pi_{\beta_j}}-\ket{\pi_{\beta_{j-1}}}\Big).
\]
Since each $U_i$ is unitary, this yields
\[
    \norm{U\ket{\pi_{\beta_l}}-\ket{\pi_{\beta_0}}}_2
    \le\sum_{j=1}^{l}\norm{U_j\ket{\pi_{\beta_j}}-\ket{\pi_{\beta_{j-1}}}}_2
    \le l\cdot\frac{\epsilon}{2l}=\frac{\epsilon}{2}.
\]
The algorithm applies $U$ to $\ket{S}$ instead of $\ket{\pi_{\beta_l}}$. Since $\beta_0=0$ we have $\pi_{\beta_0}=\pi$, and by unitarity of $U$ and the triangle inequality,
\[
    \norm{U\ket{S}-\ket{\pi}}_2
    \le\norm{U\ket{S}-U\ket{\pi_{\beta_l}}}_2+\norm{U\ket{\pi_{\beta_l}}-\ket{\pi_{\beta_0}}}_2
    \le\frac{\epsilon}{2}+\frac{\epsilon}{2}=\epsilon.
\]
Hence the output $\ket{\widetilde\pi}=U\ket{S}$ satisfies $\norm{\ket{\widetilde\pi}-\ket{\pi}}_2\le\epsilon$.

\paragraph{Q-sampling: cost.}
Because the schedule is $e^2$-slowly varying, consecutive states have constant overlap, $\braket{\pi_{\beta_j}}{\pi_{\beta_{j-1}}}\ge e^{-1}$. By \Cref{prop:qw-operator}, each walk operator has phase gap $\wtd\Omega(1/\sqrt n)$. Therefore, each $U_j$ (fixed-point search to accuracy $\epsilon/(2l)$) uses $\toh(\sqrt n\ln(l/\epsilon))$ applications of $W_{\beta_j}^{\pm1}$ and $W_{\beta_{j-1}}^{\pm1}$. Each application costs $\toh(\sqrt{m\rho/n})$ queries. Since $l\le\toh(\sqrt n)$, the total is
\[
    l\cdot\toh(\sqrt n\ln(l/\epsilon))\cdot\toh(\sqrt{m\rho/n})
    =\toh\big(\sqrt n\cdot\sqrt n\cdot\sqrt{m\rho/n}\,\ln(1/\epsilon)\big)
    =\toh\big(\sqrt{mn\rho}\cdot\ln(1/\epsilon)\big).
\]

\paragraph{Reflection.}
A reflection about $\ket{\pi}=\ket{\pi_{\beta_0}}$ is implemented by phase estimation of $W_{\beta_0}$ with precision equal to its phase gap $\wtd\Omega(1/\sqrt n)$. This takes $\toh(\sqrt n)$ applications of $W_{\beta_0}^{\pm1}$, each costing $\toh(\sqrt{m\rho/n})$ queries, for a total of $\toh(\sqrt{m\rho})$ queries.

Combining the three parts proves the theorem.
\end{proof}

  \subsection{Reusable Effective Resistance Data Structures from Sparsifiers}
  In this section we prove \Cref{prop:er-ovest}. Fix a weighted graph $G=(V,E,w)$ and a set $F\subseteq E$, and let $H$ denote the subgraph of $G$ with edge set $E\setminus F$. By a slight abuse of notation, we also write $F$ for the weighted subgraph of $G$ with edge set $F$. Since $F$ is known, oracle access to $H$ can be simulated using a single query to $\orcl$, by setting the weight of any queried edge in $F$ to zero.

The key observation is that, as $c\geq0$ varies, only the edges of $H$ are rescaled. Hence, it suffices to compute a spectral sparsifier $\wtd{H}$ of $H$ once. For every $c>0$, the graph $F+c\cdot\wtd{H}$ is then a spectral sparsifier of $G_{E\setminus F,c}$. Since both $F$ and $\wtd{H}$ are stored explicitly, this graph can be constructed for any $c\geq0$ without further queries to $\orcl$. Moreover, oracle access to $F+c\cdot\wtd{H}$ can be simulated directly from the stored representations of $F$ and $\wtd{H}$, and therefore also requires no further queries to $\orcl$.

For any $c\geq0$, we use the stored data to construct an oracle $O_{\mcl R}^{G_{E\setminus F,c}}$ acting as
\[
    O_{\mcl R}^{G_{E\setminus F,c}}
    \ket{u}\ket{v}\ket{0}
    =
    \ket{u}\ket{v}
    \ket{\tdreff^{G_{E\setminus F,c}}(u,v)},
\]
where $\tdreff^{G_{E\setminus F,c}}(u,v)$ is a $2$-effective-resistance overestimate in \mbox{$G_{E\setminus F,c}$}. Thus, after the one-time preprocessing step, such an oracle can be constructed for every graph in the family \mbox{$\{G_{E\setminus F,c}:c\geq0\}$}, and subsequently queried, without any further queries to $\orcl$.

The reusable preprocessing procedure is given in
\Cref{alg:res-eff-ds}.

\begin{algorithm}[H]
\caption{\sc Reusable effective-resistance preprocessing}
\label{alg:res-eff-ds}
\begin{algorithmic}[1]
\Require{$G=(V,E,w)$ with $n$ vertices and $m$ edges, and $F\subseteq E$.}
\Ensure{A reusable data structure for the family
$\{G_{E\setminus F,c}:c\geq0\}$.}
\State Let $H$ be the subgraph of $G$ with edge set $E\setminus F$.
\State Compute a $1/10$-spectral sparsifier $\wtd{H}$ of $H$
using \Cref{thm:spars-apersdw}.
\State Store $\wtd{H}$ and $F$.
\State \Return $\mathcal H:=(\wtd{H},F)$.
\end{algorithmic}
\end{algorithm}

We first record a property used to analyze \Cref{alg:res-eff-ds}.

\begin{lemma}[\sc Sparsifier reusability]
\label{lem:spars-reuse}
Let $G=(V,E,w)$ be a weighted graph and fix $F\subseteq E$. Let $H$ be
the subgraph of $G$ with edge set $E\setminus F$, and suppose that
$\wtd{H}$ is an $\epsilon$-spectral sparsifier of $H$. Then, for
every $c\geq0$,
\[
F+c\cdot\wtd{H}
\]
is an $\epsilon$-spectral sparsifier of $G_{E\setminus F,c}$.
\end{lemma}

\begin{proof}
Deferred to \Cref{app:sec-spars}.
\end{proof}

\begin{proposition}[\sc Reusable effective-resistance data structure]
\label{prop:er-ovest}
    Let $G=(V,E,w)$ be a weighted graph with $n$ vertices and $m$ edges,
    and let $F\subseteq E$ be stored explicitly together with its edge
    weights. For any $\eta>0$, there is a quantum algorithm that, with probability at least $1-\eta$, uses
    $\toh(\sqrt{mn}\cdot\ln(1/\eta))$ queries to the adjacency-list
    oracle $\orcl$ to construct a data structure $\mathcal H$ for the family
    of rescaled graphs $\{G_{E\setminus F,c}:c\geq0\}$.

    For any $c\geq0$, $\mathcal H$ can be used to construct an oracle
    $O_{\mathcal R}^{G_{E\setminus F,c}}$ giving access to
    $2$-effective-resistance overestimates in $G_{E\setminus F,c}$.
    Constructing this oracle, as well as querying it, requires no further
    queries to $\orcl$.
\end{proposition}

\begin{proof}
Run \Cref{alg:res-eff-ds}. Let $H$ be the subgraph of $G$ with edge set $E\setminus F$. Since $F$ is known, oracle access to $H$ can be simulated using a single query to $\orcl$, by setting the weight of any queried edge in $F$ to zero. Hence, by \Cref{thm:spars-apersdw}, with probability at least $1-\eta$, we compute a $1/10$-spectral sparsifier $\wtd{H}$ of $H$ using $\toh(\sqrt{mn}\ln(1/\eta))$ queries to $\orcl$.

Fix any $c>0$. By \Cref{lem:spars-reuse}, $F+c\cdot\wtd{H}$ is a $1/10$-spectral sparsifier of $G_{E\setminus F,c}$. Since $\wtd{H}$ and $F$ are stored explicitly, $F+c\cdot\wtd{H}$ can be constructed without further queries to $\orcl$. Moreover, oracle access to $F+c\cdot\wtd{H}$ can be simulated directly from the stored representations of $F$ and $\wtd{H}$, and hence also requires no further queries to $\orcl$.

We can therefore apply \Cref{thm:eff-res-apersdw} to the graph $F+c\cdot\wtd{H}$ with $\epsilon=1/10$. This gives, an oracle providing $2$-effective resistance overestimates in $G_{E\setminus F,c}$, using $\toh(n)$ time and no further queries to $\orcl$. Thus we can prepare
\[
O_{\mcl R}^{G_{E\setminus F,c}}
\ket{u}\ket{v}\ket{0}
=
\ket{u}\ket{v}
\ket{\tdreff^{G_{E\setminus F,c}}(u,v)}.
\]

For $c=0$, we have $G_{E\setminus F,0}=F$. Since $F$ is stored explicitly, the corresponding effective-resistance oracle can be constructed directly, without any further queries to $\orcl$.

Therefore, after the initial construction of $\mathcal H$, constructing and querying $O_{\mcl R}^{G_{E\setminus F,c}}$ for any $c\geq0$ requires no further queries to $\orcl$.
\end{proof}

  \subsection{Reusable Marginal Data Structure from }
  \label{Sec:heavy-edge}
  In this section we prove \Cref{prop:computable-heavy-envelope}. Although the marginals, or equivalently the leverage scores, vary across the graphs in the slowly varying sequence $\{\beta_i\}_i$, we show that all edges that can attain a large marginal can be identified in a single preprocessing step. This allows us to prepare the marginal state efficiently for each graph $G_{\beta}$, a key primitive for implementing the quantum walk operator in \Cref{Sec:coherentIso}.

\begin{definition}[\sc Heavy edges]
    \label{def:heavy-edges}
    Let $G=(V,E,w)$ be a weighted graph with spanning-tree distribution $\pi$, and let $p(e)=\Pr_{T\sim\pi}(e\in T)$ denote the marginal of an edge $e\in E$. For $\rho\in(0,1]$, we call $e$ \emph{$\rho$-heavy} if $p(e)\ge \rho$, or equivalently if its leverage score $\ell(e)$ satisfies $\ell(e)\ge\rho$. We denote the set of all $\rho$-heavy edges by $\heavy_\rho(G)$.
\end{definition}

Fix any edge set $F \subseteq E$. The set of $\rho$-heavy edges in $G_{E\setminus F, c}$ may vary with $c$. We now prove a monotonicity property of marginals of edges in each graph in the family $G_{E\setminus F, c}$. Consequently, the lemma shows that these sets admit a single static envelope: after storing the edges of $F$ together with the non-$F$ edges that are heavy in $G_{E\setminus F,1}$, the resulting set contains every edge that can become $\rho$-heavy in some graph in the family $\{G_{E\setminus F,c}:0<c\leq 1\}$. We refer to this set as the \emph{heavy-edge envelope}.

\begin{lemma}[\sc Marginal monotonicity under rescaling]
\label{lem:schedule-wide-heavy}
    Let $G=(V,E,w)$ be a weighted graph and fix $F\subseteq E$. For any
    $0<c_1\leq c_2\leq1$ and $e\in E$,
    \[
    p_{G_{E\setminus F,c_1}}(e)
    \geq
    p_{G_{E\setminus F,c_2}}(e)
    \quad\text{if }e\in F,
    \qquad
    p_{G_{E\setminus F,c_1}}(e)
    \leq
    p_{G_{E\setminus F,c_2}}(e)
    \quad\text{if }e\in E\setminus F.
    \]
    Consequently, for any $\rho\in[0,1]$, $\heavy_\rho({G_{E\setminus F,c_1}}) \subseteq F\cup \heavy_\rho({G_{E\setminus F,c_2}})$.
\end{lemma}

\begin{proof}
    We have $L_{G_{E\setminus F,c}}=L_F+cL_{G\setminus F}$. Let
    $0<c_1\leq c_2\leq1$.

    For $e\in E\setminus F$, $p_{G_{E\setminus F,c}}(e) =\ell_{G_{E\setminus F,c}}(e) =c\cdot w(e)\cdot b_e^\top L_{G_{E\setminus F,c}}^+b_e$. Since $L_F$ is positive semidefinite (\Cref{fact:psd}),
    \[
        x^\top(L_F+c_1L_{G\setminus F})x
        \geq
        \frac{c_1}{c_2}
        x^\top(L_F+c_2L_{G\setminus F})x.
    \]
    Since $c_1,c_2>0$, both $G_{E\setminus F,c_1}$ and
    $G_{E\setminus F,c_2}$ have the same edge support as $G$, and are therefore connected. Hence their Laplacians are positive definite on $\mathds{1}^\perp$, so the preceding quadratic-form inequality may be inverted on $\mathds{1}^\perp$. Taking $x=b_e$ gives
    \[
        b_e^\top L_{G_{E\setminus F,c_1}}^+b_e
        \leq
        \frac{c_2}{c_1}
        b_e^\top L_{G_{E\setminus F,c_2}}^+b_e.
    \]
    Multiplying by $c_1w(e)$ yields $p_{G_{E\setminus F,c_1}}(e) \leq p_{G_{E\setminus F,c_2}}(e)$.

    For $e\in F$, its weight is unchanged, and $L_F+c_1L_{G\setminus F}\leq L_F+c_2L_{G\setminus F}$. Hence, after inversion, $p_{G_{E\setminus F,c_1}}(e) \geq p_{G_{E\setminus F,c_2}}(e)$. Therefore, $\heavy_\rho(G_{E\setminus F,c_1}) \subseteq F\cup\heavy_\rho(G_{E\setminus F,c_2})$.
\end{proof}

Consequently, since $G_{E\setminus F,1}=G$, the set
$\heavy_\rho:=F\cup\heavy_\rho(G)$ contains
$\heavy_\rho(G_{E\setminus F,c})$ for every $0<c\leq1$.

The previous lemma gives a family-wide envelope defined using exact marginals, but this envelope is not directly computable. We therefore replace the $\rho$-heavy set of $G$ by the estimated heavy set obtained from effective-resistance overestimates, yielding a computable envelope with the same family-wide containment property.

\begin{lemma}[\sc Heavy set from effective-resistance overestimates]
\label{lem:heavy-set-reff-overestimates}
    Let $G=(V,E,w)$ be a connected weighted graph with $n$ vertices and $m$ edges, and let $\rho\in(n/m,1]$. Suppose we have coherent query access to $\lambda$-effective-resistance overestimates $\tdreff^G(e)$ for all $e\in E$ where $\lambda\ge1$ is an absolute constant. Define
    $$\widehat{\heavy}_\rho(G):=\{e\in E:w(e)\cdot \tdreff^G(e)\ge\rho\}.$$
    Then $\heavy_\rho(G) \subseteq\widehat{\heavy}_\rho(G)$ and
    $|\widehat{\heavy}_\rho(G)|\le \lambda(n-1)/\rho$.
    Moreover, for any \mbox{$\eta>0$}, we can compute $\widehat{\heavy}_\rho(G)$ with probability at least $1-\eta$ and uses \mbox{$\toh(\sqrt{mn/\rho}\cdot\ln(1/\eta))$} queries to the adjacency-list oracle $\orcl$.
\end{lemma}

\begin{proof}
    For any $e\in E$, $\ell(e)=w(e)\reff^G(e)$. Since
    $\tdreff^G(e)\geq\reff^G(e)$, every $e\in\heavy_\rho(G)$ satisfies $w(e)\tdreff^G(e)\geq\ell(e)\geq\rho$, and hence $\heavy_\rho(G)\subseteq\widehat{\heavy}_\rho(G)$.

    Moreover,
    \[
        \rho|\widehat{\heavy}_\rho(G)|
        \leq
        \sum_{e\in\widehat{\heavy}_\rho(G)}w(e)\tdreff^G(e)
        \leq
        \lambda\sum_{e\in E}\ell(e)
        =
        \lambda(n-1),
    \]
    where the last equality follows from \Cref{fact:foster}.

    Finally, we apply repeated Grover search over the $2m$ adjacency-list positions. Each undirected edge $e=\{u,v\}$ appears twice, once as a position $(u,i)$ in the adjacency list of $u$ and once as a position $(v,j)$ in the adjacency list of $v$. For a position $(u,i)$, one query to $\orcl$ returns the corresponding edge $e=\{u,v\}$ and its weight, and we mark the position if $w(e)\tdreff^G(e)\geq\rho$. Hence every edge in $\widehat{\heavy}_\rho(G)$ contributes exactly two marked positions. \Cref{lem:repeated-grover} finds all such positions with probability at least $1-\eta$ using
    \[
        \toh\left(
        \sqrt{2m \cdot \frac{2\lambda(n-1)}{\rho}}
        \cdot\ln(1/\eta)
        \right)
        =
        \toh\left(
        \sqrt{\frac{mn}{\rho}}
        \cdot\ln(1/\eta)
        \right)
    \]
    queries to the adjacency-list oracle $\orcl$.
\end{proof}

We now show that once can use this heavy-edge envelope to prepare the marginal-state efficiently.

\begin{lemma}[\sc Marginal-state preparation]
\label{lem:marginal-state-preparation}
Let $G=(V,E,w)$ be a connected weighted graph with $n$ vertices and
$m$ edges, and let $\rho\in[n/m,1]$. The marginal state $\ket{\mu_G}$
corresponding to
$\wtd p_G(e):=w(e)\tdreff^G(e)$ can be prepared using
$\toh(\sqrt{m\rho/n})$ queries to $\mcl O_G$, provided that we are
given:
\begin{itemize}
    \item coherent query access to $2$-effective-resistance overestimates $\tdreff^G(e)$ for all edges of $G$;
    \item the degree of every vertex of $G$; and
    \item a stored set $\mcl E_\rho\subseteq E$ such that every edge outside $\mcl E_\rho$ has marginal less than $\rho$, with each edge in $\mcl E_\rho$ stored together with its weight and its two adjacency-list positions.
\end{itemize}
\end{lemma}

\begin{proof}
    From the given vertex degrees, we compute the prefix sums $D_u:=\sum_{v<u}\deg(v)$ with respect to the fixed lexicographic order on the vertices. We label the $2m$ adjacency-list positions by assigning $k:=D_u+i$ to $(u,i)$. Conversely, given $k\in[2m]$, $u$ is the unique vertex satisfying $D_u<k\leq D_u+\deg(u)$, and $i=k-D_u$. Thus we can implement
    \[
        \ket{k}
        \xmapsto{\rm prefix}
        \ket{u}\ket{i}
        \xmapsto{\orcl}
        \ket{u}\ket{i}\ket{v}\ket{w(u,v)},
    \]
    where $v$ is the $i$-th neighbor of $u$. For an edge $e=\{u,v\}$, among its two adjacency-list positions $(u,i)$ and $(v,j)$ we use the one whose first vertex is smaller in the fixed lexicographic order.

    Consider the $2$-marginal overestimate $\wtd p_{G}(e):=w_{G}(e)\tdreff^{G}(e)$ and the corresponding isotropic multiplicities $t_{G}(e)$ from \Cref{def:isotropic-multiplicity}. Writing $M:=\sum_{e\in E}t_{G}(e)$, the marginal state (cf \Cref{def:marginal-state}) is
    \[
        \ket{\mu_{G}}
        =
        \frac{1}{\sqrt M}
        \sum_{e\in E}\sqrt{t_{G}(e)}\,\ket e.
    \]
    Let $p_{G}(e)$ denote the marginal of $e$ under the spanning-tree distribution of $G$. If $e\notin\mcl E_\rho$, then $p_{G}(e)<\rho$. Since $\tdreff^{G}$ is a $2$-effective-resistance overestimate (cf \Cref{def:marg-res-ovest}),  
    \[
        \wtd p_{G}(e)\leq2w_{G}(e)\reff^{G}(e)
        =2p_{G}(e)<2\rho.
    \] 
    Hence $t_{G}(e)\leq t_{\rm cap}$, where $t_{\rm cap}:=\left\lceil m\rho/(n-1)\right\rceil$. Define the proposal weights
    \[
        q(e):=
        \begin{cases}
            t_{G}(e), & e\in\mcl E_\rho,\\
            t_{\rm cap}, & e\notin\mcl E_\rho.
        \end{cases}
    \]
    Then $q(e)\geq t_{G}(e)$ for every $e\in E$.

    Since $\wtd p_{G}$ is a $2$-marginal overestimate and
    $\lceil x\rceil\leq1+x$, $m\leq M\leq m+\frac{m}{2(n-1)}\sum_{e\in E}\wtd p_{G}(e)\leq2m$. Let $Q:=\sum_{e\in E}q(e)$. Using $\rho\geq n/m$,
    \[
        Q = \sum_{e\in \mcl{E}_{\rho}}t_{G}(e) + \sum_{e\notin \mcl{E}_{\rho}}t_{\rm cap} \le \sum_{e\in E}t_{G}(e)+m\cdot t_{\rm cap} \le 2m+m\left\lceil \frac{m\rho}{n-1}\right\rceil = O\left(\frac{m^2\rho}{n}\right).
    \]

    We next use the proposal weights to prepare a proposal state $\ket{q}$. Construct a binary tree over the $2m$ edge labels, initially assigning $t_{\rm cap}$ to every leaf. For each $e\in\mcl E_\rho$, its two adjacency-list positions $(u,i)$ give its two labels $k=D_u+i$ without queries to $\orcl$. We replace the values at these two leaves by $t_{G}(e)$, which is computable from the stored edge weight and $\tdreff^{G}(e)$. Thus the leaf corresponding to a label $k$ has value $q(e(k))$, where $e(k)$ is the edge at that adjacency-list position. This can be computed without any queries to $\orcl$ since we already have stored the weights of edges in $\mcl E_{\rho}$ and we have coherent query access to $2$-effective-resistance overestimates $\tdreff^{G}(e)$. The sum of all leaf values is thus $2Q$. Storing the corresponding partial sums at the internal nodes, the recursive state-preparation procedure of Grover and Rudolph~\cite{GR02p} prepares
    \[
        \frac{1}{\sqrt{2Q}}
        \sum_{k\in[2m]}
        \sqrt{q(e(k))}\ket{k}
    \]
    without queries to $\orcl$.

    Using the stored prefix sums, we first map each edge label to its adjacency-list position:
    \[
        \frac{1}{\sqrt{2Q}}
        \sum_{k\in[2m]}\sqrt{q(e(k))}\ket{k}
        \longmapsto
        \frac{1}{\sqrt{2Q}}
        \sum_{(u,i)}
        \sqrt{q(e)}\ket{u,i},
    \]
    where $e$ is the edge at position $(u,i)$. The label $k$ is uncomputed using $k=D_u+i$. We then query $\orcl$ to obtain the neighbor $v$ at position $(u,i)$ and mark the terms with $u<v$. Since every edge occurs at exactly two adjacency-list positions with equal weight and exactly one is marked, the success probability is $1/2$. By a slight abuse of notation, we use $\ket e$ to denote the unique marked basis state corresponding to each edge $e$. Hence constant-overhead amplitude amplification~\cite{BHMT02j} prepares
    \[
        \ket q = \frac{1}{\sqrt Q} \sum_{e\in E}\sqrt{q(e)}\,\ket e.
    \]

    Starting from $\ket q$, we perform quantum rejection sampling as in \cite{ORR13j}. For a basis state $\ket e$, one query to $\mcl O_G$ at its corresponding adjacency-list position gives the other endpoint and the edge weight $w(e)$. Together with the given $\tdreff^G(e)$, this allows us to compute $t_G(e)$; the proposal weight $q(e)$ is obtained from the proposal tree without queries to
    $\mcl O_G$. Since $q(e)\geq t_G(e)$ for every edge, using $O(1)$ queries we perform
    \[
        \ket q
        \xmapsto{\rm rot.}
        \frac{1}{\sqrt Q}
        \sum_{e\in E}\sqrt{q(e)}\,\ket e
        \left(
            \sqrt{\frac{t_{G}(e)}{q(e)}}\ket1
            +
            \sqrt{1-\frac{t_{G}(e)}{q(e)}}\ket0
        \right).
    \]
    The auxiliary information used to implement the rotation is then
    uncomputed using $\toh(1)$ queries to $\mcl O_G$. The component marked by $\ket1$ is
    \[
        \frac{1}{\sqrt Q}
        \sum_{e\in E}\sqrt{t_G(e)}\,\ket e\ket1 \propto \ket{\mu_G},
    \]
    which has squared norm $M/Q$. Hence amplitude amplification~\cite{BHMT02j} gives
    \[
        \frac{1}{\sqrt Q}
        \sum_{e\in E}\sqrt{t_G(e)}\,\ket e\ket1
        \xmapsto{\rm AA}
        \frac{1}{\sqrt M}
        \sum_{e\in E}\sqrt{t_G(e)}\,\ket e\ket1
        =
        \ket{\mu_G}\ket1
    \]
    using
    \[
        O\left(\sqrt{\frac QM}\right)
        =
        O\left(\sqrt{\frac{m\rho}{n}}\right)
    \]
    iterations. Each iteration uses $O(1)$ queries to $\mcl O_G$, while $\tdreff^G$ is available coherently by assumption. Hence $\ket{\mu_G}$ is prepared using $\toh(\sqrt{m\rho/n})$ queries to $\mcl O_G$.
\end{proof}

\begin{proposition}[\sc Reusable marginal data structure]
\label{prop:computable-heavy-envelope}
    Consider any weighted graph $G=(V,E,w)$ with $n$ vertices and $m$ edges, and let $F\subseteq E$ be stored explicitly together with its edge weights. For any $\rho\in[n/m,1]$ and $\eta>0$, there is a quantum algorithm that constructs a data structure $\margds$ providing coherent query access to a set of edges, together with their weights, such that every edge $e\in E$ outside this set has marginal less than $\rho$ in every graph in the family $\{G_{E\setminus F,c}:0<c\leq1\}$. The algorithm succeeds with probability at least $1-\eta$ and uses $\toh(\sqrt{mn/\rho}\cdot\ln(1/\eta))$ queries to the adjacency-list oracle $\orcl$.

    Moreover, after constructing $\margds$, for any $0<c\leq1$ it can be used to prepare the marginal state $\ket{\mu_{G_{E\setminus F,c}}}$ corresponding to the $2$-marginal overestimate $\wtd p_{G_{E\setminus F,c}}(e):=w_{G_{E\setminus F,c}}(e)\tdreff^{G_{E\setminus F,c}}(e)$ (\Cref{def:marginal-state}), using $\toh(\sqrt{m\rho/n})$ queries to $\orcl$, given coherent query access to $2$-effective-resistance overestimates $\tdreff^{G_{E\setminus F,c}}$.
\end{proposition}

\begin{proof}
    Let $\mcl E_\rho:=\widehat{\heavy}_\rho(G)\cup F$. By
    \Cref{lem:schedule-wide-heavy} and $\heavy_\rho(G)\subseteq\widehat{\heavy}_\rho(G)$, for every $0<c\leq1$ we have $\heavy_\rho(G_{E\setminus F,c})\subseteq\mcl E_\rho$. Hence every edge in $E\setminus\mcl E_\rho$ has marginal less than $\rho$ in every graph in the family
    $\{G_{E\setminus F,c}:0<c\leq1\}$.

    Since $F$ is already stored, the query bound for computing $\mcl E_\rho$ follows from \Cref{lem:heavy-set-reff-overestimates}. By the implementation in its proof, both adjacency-list positions and the edge weight of every edge in $\widehat{\heavy}_\rho(G)$ are obtained during the Grover search.

    After constructing $\margds$, for every $0<c\leq1$, together with the assumed coherent query access to the corresponding $2$-effective-resistance overestimates, the inputs of \Cref{lem:marginal-state-preparation} are available for $G_{E\setminus F,c}$. Hence $\ket{\mu_{G_{E\setminus F,c}}}$ can be prepared using $\toh(\sqrt{m\rho/n})$ queries to $\orcl$.
\end{proof}

  \subsection{Coherent Isotropization}
  \label{Sec:coherentIso}
  In this section we prove \Cref{prop:qw-operator}. The goal is to implement, for any weighted graph $G=(V,E,w)$, a quantum walk operator $W$ with large phase gap whose unique $1$-eigenvector is $\ket{\pi}$, the q-sample of the spanning-tree distribution $\pi$ of $G$. This allows us to implement reflections about $\ket{\pi}$ efficiently.

We quantize the Up-Down walk with transitions biased according to the marginals as described in \Cref{mc:marginal-aware}. More specifically, we use the isotropic transformation from \Cref{prop:iso-transform} only implicitly. Given marginal overestimates, let $t_G(e)$ denote the corresponding isotropic multiplicity of each edge $e$. Following \Cref{prop:iso-transform}, we call a spanning tree $T'$ of $G'$ a copied tree of $T$ if its edges are copies of the edges of $T$, and say that $T'$ projects to $T$. For an original spanning
tree $T$, define its copied-tree weight by
\[
    \nu(T):=
    \frac{\pi(T)}{\prod_{e\in T}t_G(e)}.
\]
Thus every copied tree $T'$ projecting to $T$ satisfies
$\pi'(T')=\nu(T)$.

We define an equivalent walk on the original spanning trees.  Let
\[
M:=\sum_{f\in E}t_G(f),
\qquad
Z_T(f):=\sum_{e\in\Cycle(T+f)}\nu(T+f-e).
\] 
The up-step samples edge $f \in E$ with probability proportional to $t_G(f)$, while the down-step removes an edge $e \in \Cycle(T+f)$ with probability proportional to the copied-tree weights $\nu(T+f-e)$. For simplicity, we take $\Cycle(T+f)=\{f\}$ when $f\in T$, so that the down-step is a self-loop in this case. Working directly on the original spanning trees keeps the state space and its basis encoding independent of the edge copies introduced by the isotropic transformation.

\begin{markovchain}[H]
\caption{Marginal-aware up-down walk from $T_0$}
\label{mc:marginal-aware}
\begin{algorithmic}[1]
\For{$i=0,1,2,\ldots$}
    \State Choose $f\in E$ with probability $t_G(f)/M$.
    \If{$f\in T_i$}
        \State Set $T_{i+1}=T_i$.
    \Else
        \State Choose $e\in\Cycle(T_i+f)$ with probability
        $\nu(T_i+f-e)/Z_{T_i}(f)$.
        \State Set $T_{i+1}=T_i+f-e$.
    \EndIf
\EndFor
\end{algorithmic}
\end{markovchain}

The down step is exactly the down step of the Up-Down walk on the isotropically transformed graph, expressed on the original edges.

\begin{proposition}[\sc Detailed balance of the marginal-aware walk] \label{prop:ma-detailed-balance}
    The marginal-aware walk is reversible with respect to $\pi$.
\end{proposition}

\begin{proof}
    Let $P_{\rm ma}$ be the transition matrix, and consider neighboring trees $T$ and $S=T+f-e$. Since $T+f=S+e$, we have $Z_T(f)=Z_{S}(e)$. Therefore
    \[
        P_{\rm ma}(T,S)
        =
        \frac{t_G(f)}{M}\frac{\nu(S)}{Z_T(f)},
        \qquad
        P_{\rm ma}(S,T)
        =
        \frac{t_G(e)}{M}\frac{\nu(T)}{Z_T(f)}.
    \]
    Moreover, since $S$ is obtained from $T$ by replacing $e$ with $f$,
    \[
        \frac{\nu(S)}{\nu(T)}
        =
        \frac{\pi(S)}{\pi(T)}
        \frac{t_G(e)}{t_G(f)}.
    \]
    Hence $\pi(T)t_G(f)\nu(S)=\pi(S)t_G(e)\nu(T)$, and dividing by the common factor $MZ_T(f)$ gives
    $\pi(T)P_{\rm ma}(T,S)=\pi(S)P_{\rm ma}(S,T)$.
    The diagonal case is immediate, so detailed balance holds.
\end{proof}

We will also use the following standard fact about the spectra of Markov chains under lumping or projecting based on equivalence relation defined on the states of the Markov chain.
\begin{lemma}[\sc Spectrum of Lumped Chain; {\cite[Lemma 12.9]{LP17b}}]
    \label{lem:lumped-spectra}
    Let $\mathcal{X}$ be the state space of a Markov chain $\left(X_t\right)$ with transition matrix $P$. Let $\sim$ be an equivalence relation on $\mathcal{X}$ with equivalence classes $\mathcal{X}^{\sharp}= \{[x]: x \in \mathcal{X}\}$ such that $\left[X_t\right]$ is a Markov chain with transition matrix $P^{\sharp}([x],[y])= P(x,[y])$. Then:
    \begin{enumerate}
        \item Let $f: \mathcal{X} \rightarrow \mathbb{R}$ be an eigenfunction of $P$ with eigenvalue $\lambda$ which is constant on each equivalence class. Then the natural projection $f^{\sharp}: \mathcal{X}^{\sharp} \rightarrow \mathbb{R}$ of $f$, defined by $f^{\sharp}([x])=f(x)$, is an eigenfunction of $P^{\sharp}$ with eigenvalue $\lambda$.
        \item Conversely, if $g: \mathcal{X}^{\sharp} \rightarrow \mathbb{R}$ is an eigenfunction of $P^{\sharp}$ with eigenvalue $\lambda$, then its lift $g^{\flat}: \mathcal{X} \rightarrow \mathbb{R}$, defined by $g^{\flat}(x)=g([x])$, is an eigenfunction of $P$ with eigenvalue $\lambda$.
    \end{enumerate}
\end{lemma}

\begin{lemma}[\sc Spectral gap of the marginal-aware walk]
\label{lem:margaw-walk-gap}
Let $G=(V,E,w)$ be a connected weighted graph on $n$ vertices and $m$ edges,
with $m\ge 100n$. Let $G'=(V,E',w')$ be the isotropic transformed multigraph (cf \Cref{prop:iso-transform}),
and suppose that its spanning tree distribution $\pi'$ is nearly isotropic. Then the marginal-aware walk \Cref{mc:marginal-aware} satisfies
\[
    \gamma(P_{\rm ma})
    \ge
    \Omega\left(\frac{1}{n\log^2 n}\right).
\]
\end{lemma}

\begin{proof}
Let $Q$ be the up-down walk on the isotropic transformed multigraph $G'$. Since $\pi'$ is nearly isotropic, \Cref{thm:up-down-gap} gives $\gamma(Q)\ge \Omega(1/(n\log^2 n))$. We project a copied tree in $G'$ to the original tree in $G$ obtained by forgetting copy indices. For an original tree $T$, let
\[
    [T]:=\{T':T'\text{ projects to }T\}
\]
be the corresponding equivalence class of copied trees. This projection defines a Markov chain $Q^\sharp$ on original trees. For $S=T+f-e$ and any $T'\in[T]$,
\[
    Q^\sharp(T,S):=Q(T',[S])=\frac{t_G(f)}{M-(n-1)}\frac{\nu(S)}{Z_{T}(f)}.
\]

The right-hand side is independent of the choice of $T'\in[T]$, so the projected chain is well-defined. By \Cref{lem:lumped-spectra}, projecting in this way cannot decrease the spectral gap, hence we get $\gamma({Q^\sharp})\ge\gamma(Q)$. The stationary distribution of $Q^\sharp$ is $\pi$. Indeed, for every
original tree $T$,
\[
    \sum_{T'\in[T]}\pi'(T')
    =
    \prod_{e\in T}t_G(e)\,\nu(T)
    =
    \pi(T).
\]

Now compare $Q^\sharp$ with the marginal-aware walk $P_{\rm ma}$. For $S=T+f-e$,
\[
    P_{\rm ma}(T,S)=\frac{t_G(f)}{M}\frac{\nu(S)}{Z_{T}(f)}=\zeta \cdot Q^\sharp(T,S),
    \qquad
    \zeta:=\frac{M-(n-1)}{M}.
\]
Since $P_{\rm ma}=\zeta Q^\sharp+(1-\zeta)I$, every eigenvalue $\lambda$ of $Q^\sharp$ is mapped to $1-\zeta+\zeta\lambda$ in $P_{\rm ma}$. Hence $\gamma(P_{\rm ma})=\zeta\cdot \gamma(Q^{\sharp})$. Since $M\ge m\ge 100n$, we have $\zeta\ge 99/100$, and therefore
\[
    \gamma(P_{\rm ma})\ge \frac{99}{100}\cdot \gamma(Q^{\sharp})\ge \Omega\left(\frac{1}{n\log^2 n}\right).
\]
\end{proof}

\begin{lemma}[\sc Nonnegative spectrum of the marginal-aware walk]
\label{lem:margaw-psd}
The marginal-aware walk has nonnegative spectrum. Consequently, its
absolute spectral gap equals its ordinary spectral gap.
\end{lemma}

\begin{proof}
The up-down walk $Q$ on the isotropically transformed multigraph has
nonnegative spectrum \cite{CSV23p}. By
\Cref{lem:lumped-spectra}, every eigenvalue of $Q^\sharp$ is also an
eigenvalue of $Q$, and hence $Q^\sharp$ has nonnegative spectrum.
Since
$P_{\rm ma}=\zeta Q^\sharp+(1-\zeta)I$, every eigenvalue $\lambda$ of
$Q^\sharp$ is mapped to $1-\zeta+\zeta\lambda\geq0$. Thus
$P_{\rm ma}$ has nonnegative spectrum, and its absolute spectral gap
equals its ordinary spectral gap.
\end{proof}

\paragraph{Quantization of the marginal-aware walk.}
We now express the marginal-aware walk in the labelled Szegedy framework of \Cref{thm:labelled-szegedy}. For each spanning tree $T$, define the set of transition labels
\[
    C_{T} := \{(f,e):f\notin T,\ e\in\Cycle(T+f)\} \cup \{(f,f):f\in T\}.
\]
For $(f,e)\in C_{T}$, probability of making a transition from $T$ along the label $(f,e)$ is 
\[
\Pr[(f,e) \mid T]
:=
\begin{cases}
\dfrac{t_G(f)}{M},
& f\in T,\ e=f,\\[1ex]
\dfrac{t_G(f)}{M}
\dfrac{\nu(T+f-e)}{Z_{T}(f)},
& f\notin T,\ e\in\Cycle(T+f).
\end{cases}
\]
The endpoint map is $h_T(f,e):=T+f-e$, with
$T+f-f:=T$. Hence we have 
\[
    P_{\rm ma}(T,S)=\sum_{(f,e):\,h_T(f,e)=S}\Pr[(f,e) \mid T].
\]

For every non-self-loop transition $S=T+f-e$, the exchanged edges are uniquely determined by $f=S\setminus T$ and $e=T\setminus S$. Define the flip-flop shift by
\[
    g\bigl(T,(f,e)\bigr) := \bigl(T+f-e,(e,f)\bigr).
\]
This is an involution, and every label with $e=f$ is fixed. Therefore the marginal-aware walk satisfies the labelled-transition assumptions of \Cref{thm:labelled-szegedy}.

\begin{proposition}[\sc Quantum walk operator with large gap]\label{prop:qw-operator}
    Let $G=(V,E,w)$ be any weighted graph with $n$ vertices and $m\ge 100n$ edges, and let $\Delta$ denote the set of spanning trees of $G$. There exists a quantum walk operator $W$, associated with a reversible Markov chain on $\Delta$, whose unique $1$-eigenvector is $\ket{\pi}$, where $\pi$ is the spanning-tree distribution of $G$, and whose phase gap is at least $\wtd{\Omega}(1/\sqrt n)$.

    Moreover, for any $\rho\in[n/m,1]$, one application of $W$ and its inverse can be implemented using $\toh(\sqrt{m\rho/n})$ queries to $\mcl O_G$, given:
    \begin{itemize}
        \item coherent query access to $2$-effective-resistance overestimates for all edges of $G$;
        \item the degree of every vertex of $G$; and
        \item a stored set of edges such that every edge outside this set has marginal under $\pi$ less than $\rho$, with each stored edge given together with its weight and its two adjacency-list positions.
    \end{itemize}
\end{proposition}

\begin{proof}
Let $\mcl{M}_{\rm ma}$ be the marginal-aware walk from \Cref{mc:marginal-aware}, and let $P_{\rm ma}$ be its transition matrix. By \Cref{prop:ma-detailed-balance}, $\mcl{M}_{\rm ma}$ is reversible with stationary distribution $\pi$. By \Cref{lem:margaw-walk-gap,lem:margaw-psd}, we have $\delta_{\rm abs}(P_{\rm ma})\ge \wtd{\Omega}(1/n)$.

The exchange-pair labels described before give a labelled-transition realization of $\mcl{M}_{\rm ma}$. Therefore, \Cref{thm:labelled-szegedy} applies to the corresponding labelled Szegedy walk operator $W$. On the busy subspace, its unique $1$-eigenvector is $\ket{\pi}\ket{0}$. Since the phase gap is $\Omega(\sqrt{\delta_{\rm abs}(P_{\rm ma})})$, we obtain
\[
\gamma_{\rm ph}(W)\ge \wtd{\Omega}(1/\sqrt n).
\]

It remains to bound the implementation cost of one application of $W$. One application consists of preparing and unpreparing the labelled transition superposition $\mcl U$ and applying the flip-flop shift $\mcl S$ (see \Cref{sec:quantum-walk}). The reflection about the initialized label register and the flip-flop shift themselves use no adjacency-list queries, so it suffices to implement the transition preparation.

Recall that $\mcl U$ acts as
\[
\mcl U\ket{T}\ket{0}
\mapsto
\ket{T}
\sum_{(f,e)\in C_T}
\sqrt{\Pr[(f,e)\mid T]}\ket{f}\ket{e}.
\]

We implement $\mathcal U$ in two stages:

\begin{align*}
    \ket{T}\ket{0} &\xrightarrow{\text{(1)}} \ket{T}\ket{\mu_G} = \ket{T}\frac{1}{\sqrt{M}} \sum_{f\in E}\sqrt{t_G(f)}\ket{f}\\
    &\xrightarrow{\text{(2)}}\ket{T}\frac{1}{\sqrt{M}}
    \sum_{f\in E}\sqrt{t_G(f)}\ket{f}\frac{1}{\sqrt{Z_T(f)}}\sum_{e\in\Cycle(T+f)}\sqrt{\nu(T+f-e)}\ket{e}\\
    &= \ket{T}\sum_{(f,e)\in C_T}\sqrt{\Pr[(f,e)\mid T]}\ket{f}\ket{e}.
\end{align*}

Here, step (1) prepares the up-step, while step (2) prepares the down-step conditioned on $T$ and $f$. We describe the implementation of these two steps separately below.

Step (1) is exactly the preparation of $\ket{\mu_G}$. The assumptions of the proposition satisfy all the requirements of \Cref{lem:marginal-state-preparation}, and thus we can perform step (1) using $\wtd{O}(\sqrt{m\rho/n})$ queries to $\orcl$.

Next, we implement step (2), namely, the down-step. Given the current tree $T$ and the added edge $f$, the unique cycle $\Cycle(T+f)$ can be found from the stored representation of $T$, without any further queries to $\mcl O_G$. This can be done naively by traversing $T$, whose explicit description is already stored in $\ket{T}$.

On this cycle, the down-step probabilities are proportional to $\nu(T+f-e)$, equivalently to $t_G(e)/w(e)$ for $e\in\Cycle(T+f)$. Thus, it suffices to prepare the down-step superposition using weights proportional to $t_G(e)/w(e)$.

Given the current tree $T$ and the added edge $f$, the unique cycle $\Cycle(T+f)$ is determined from the stored representation of $T$, without any further queries to $\mcl O_G$. We construct a binary tree whose leaves correspond to the edges $e\in\Cycle(T+f)$, assigning to the leaf corresponding to $e$ the value $t_G(e)/w(e)$. The value $t_G(e)$ can be computed from the coherent query access to the $2$-effective-resistance overestimate for $e$, while $w(e)$ is available from the stored edge weight. Thus, all leaf values can be computed without queries to $\mcl O_G$. Storing the corresponding partial sums at the internal nodes, the recursive state-preparation procedure of Grover and Rudolph~\cite{GR02p} prepares
\[
\frac{1}{\sqrt{Z_T(f)}}
\sum_{e\in\Cycle(T+f)}
\sqrt{\nu(T+f-e)}\ket{e}
\]
conditioned on $\ket{T}\ket{f}$, since the leaf weights are proportional to $\nu(T+f-e)$. Hence, step (2) requires no additional queries to $\mcl O_G$.

The remaining part of the implementation of $W$ is the flip-flop shift, which maps
\[
(T,f,e)\mapsto (T+f-e,e,f),
\]
fixing $(T,f,f)$ when $f\in T$. This can be implemented by swapping the two edge registers and updating the tree register, all without any further queries to $\mcl O_G$. Therefore, one application of $W$ can be implemented using $\wtd{O}(\sqrt{m\rho/n})$ queries to $\mcl O_G$.
\end{proof}

  \subsection{Slowly Varying Schedule for Spanning Trees}
  \label{Sec:twoCS}
  In this section, we will prove \Cref{prop:compute-adapt}. In order to prove it, we need some properties about the cooling schedule framework of spanning trees. Consider $G=(V,E,w)$, a fixed weighted graph with $n$ vertices and $m$ edges, and let $S \subseteq E$ be a fixed \emph{admissible} spanning tree of $G$.

Recall from \Cref{def:beta-scale} that, after fixing an admissible spanning tree $S$, the distribution $\pi_\beta^S$ is defined by $\pi_\beta^S(T)\propto e^{-\beta |T\setminus S|}w(T)$, so that $\pi_0^S=\pi$ and $\pi_\infty^S=\mathds{1}_S$.


When $\beta = \infty$,  We approximate $\pi_{\infty}$ by $\pi_{\beta}$ where $\beta$ is a smaller finite value which is emphasised in the following lemma similar to Proposition 21 of \cite{AD20c}.


\begin{lemma}[\sc Closeness to point mass]
\label{lem:pointmass}
    Fix $S$, an admissible spanning tree of $G$. Let $\pi_{\beta^{\star}}^S$ be the spanning tree distribution of the weighted graph
    $G_{\beta^{\star}}=(V,E,w_{\beta^{\star}})$ as in \Cref{def:beta-scale}, where $\beta^{\star} = \ln\left({16n^{n-2}}/{\epsilon^2}\right)$.
    Then \mbox{$\norm{\pi_{\beta^{\star}}^S-\pi_\infty^S}_{\mathrm{TV}} \le {\epsilon^2}/{8}$}. Consequently, $\norm{\ket{\pi_{\beta^{\star}}^S}-\ket{\pi_\infty^S}}_2 \leq \epsilon/2$.
\end{lemma}

\begin{proof}
    Since $\pi_\infty^S=\mathds{1}_S$, it is enough to bound
    $\norm{\pi_{\beta^{\star}}^S-\mathds{1}_S}_{\mathrm{TV}}$. For a point mass, we have that $\norm{\pi_{\beta^{\star}}^S-\mathds{1}_S}_{\mathrm{TV}} = 1-\pi_{\beta^{\star}}^S(S)$.
    By definition of the $\beta$-scaled distribution,
    \[
        1-\pi_{\beta^{\star}}^S(S)
        =
        \frac{\sum_{T\neq S} e^{-\beta^{\star} |T\setminus S|} w(T)}{w(S)+\sum_{T\neq S} e^{-\beta^{\star} |T\setminus S|} w(T)} \leq 
        \frac{e^{-\beta^{\star}}\sum_{T\neq S} w(T)}{w(S)}
        \leq
        \frac{e^{-\beta^{\star}}}{\pi(S)}.
    \]

    Using admissibility of $S$, namely $\pi(S)\geq {1}/{2n^{n-2}}$, we obtain
    \[
        1-\pi_{\beta^{\star}}^S(S) \leq 2n^{n-2} e^{-\beta^{\star}} = 2n^{n-2}\cdot \frac{\epsilon^2}{16n^{n-2}} = \frac{\epsilon^2}{8}.
    \]

    Finally, for any probability distributions $p,q$,
    \[
        \norm{\ket p-\ket q}_2^2 =\sum_x(\sqrt{p(x)}-\sqrt{q(x)})^2 \leq 2\norm{p-q}_{\rm TV}.
    \] 
    Therefore,$\norm{\ket{\pi_{\beta^\star}^S}-\ket{\pi_\infty^S}}_2 \leq \sqrt{2\cdot\frac{\epsilon^2}{8}} ={\epsilon}/{2}$.
\end{proof}

To analyze the path in $\beta$, we use a cooling schedule from $\beta^{\star}$ down to $0$. Outputting $S$ gives an initial approximation to the q-sample for $\pi_{\beta^{\star}}^S$, and the annealing procedure moves from $\beta=\beta^{\star}$ back to $\beta=0$, where the target distribution is $\pi$. The tree $S$ is used only to define and initialize this annealing path. It remains to show that such an admissible spanning tree can be found efficiently. The following proposition gives the required preprocessing step; its proof is deferred to \Cref{app:sec-cooling-sch}.

\begin{proposition}[\sc Computing an Admissible Tree]\label{prop:sample-admissible}
    Let $G(V,E,w)$ be a weighted, connected graph with $\pi$ as the spanning tree distribution of $G$. There is a quantum algorithm that, for any $\eta>0$, outputs an admissible spanning tree $S$ with probability $1-\eta$ using $\toh(\sqrt{mn}\cdot\log(1/\eta))$ adjacency-list queries and time. 
\end{proposition}

\begin{proof}
    Let $S$ be a maximum-weight-product spanning tree of $G$, so that $\prod_{e\in S}w(e)\geq\prod_{e\in T}w(e)$ for every $T\in\Delta$. Hence,
    \[
    \pi(S)
    =
    \frac{\prod_{e\in S}w(e)}
    {\sum_{T\in\Delta}\prod_{e\in T}w(e)}
    \geq \frac{1}{|\Delta|}
    \geq \frac{1}{n^{n-2}},
    \]
    where the last inequality follows from the fact that an $n$-vertex graph has at most $n^{n-2}$ spanning trees. Thus, $S$ is admissible. By \Cref{thm:max-product}, we can find such a tree with probability at least $1-\eta$ using $\toh(\sqrt{mn}\log(1/\eta))$ adjacency-list queries and time.
\end{proof}

The remaining task is to choose a short slowly-varying schedule $0=\beta_0<\beta_1<\cdots<\beta_l=\beta^{\star}$. 

We now instantiate the preceding Gibbs framework for spanning trees (cf \Cref{sec:sim-anl}). Fix an admissible spanning tree $S$, let $\Delta$ be the set of spanning trees of $G$, set the base weight to be $\mu(T)=w(T)$, and define the Hamiltonian $H_S(T):=|T\setminus S|$. Then the partition function is $Z_S(\beta)=\sum_{T\in\Delta}e^{-\beta H_S(T)}w(T)$, and the corresponding Gibbs distribution is exactly the $\beta$-scaled spanning tree distribution from \Cref{def:beta-scale}: $\pi_\beta^S(T)=e^{-\beta H_S(T)}w(T)/Z_S(\beta)$.

\begin{lemma}[\sc Log-convexity of the partition function]
\label{lem:log-convexity-partition}
For any fixed admissible spanning tree $S$, the partition function
$Z_S(\beta)$ is log-convex in $\beta$.
\end{lemma}
\begin{proof}
    The proof is deferred to \Cref{app:sec-cooling-sch}.
\end{proof}

In our application, we start from an admissible tree, corresponding to the large-$\beta$ endpoint $\ket{\pi_\infty}$, and move towards the target graph at smaller $\beta$. Since \Cref{thm:q-adapt-anneal} assumes an efficiently preparable state at $\beta=0$, we use the following reversed version to construct the schedule starting from the easily preparable large-$\beta$ endpoint. Although the proposition follows from existing results, we state it here for completeness and to keep the argument self-contained.

\begin{proposition}[\sc Quantum algorithm for a reversed slowly-varying schedule, adapted from~\cite{HW20c}]
\label{prop:q-adapt-anneal-reverse}
    Let $\{\pi_\beta\}_{\beta\in[0,\beta^{\star}]}$ be a family of Gibbs distributions with log-convex partition function $Z(\beta)$. There exists an $e^2$-slowly varying schedule $0=\beta_0<\cdots<\beta_l=\beta^{\star}$ of length $l \leq \toh\left(\sqrt{\ln{\frac {Z(0)}{Z(\beta^{\star})}}}\right)$.

    Moreover, if a state $\ket{\wtd{\pi}_{\beta^\star}}$ satisfying $\norm{\ket{\wtd{\pi}_{\beta^\star}}-\ket{\pi_{\beta^\star}}}_2\leq\epsilon$ can be prepared and, for every $\beta\in[0,\beta^\star]$, there is a reversible Markov chain $\mcl{M}_\beta$ with stationary distribution $\pi_\beta$ and absolute spectral gap at least $\delta$, together with an implementable quantum walk operator $W_\beta$, then for any $\eta>0$, a quantum algorithm constructs such a schedule with probability at least $1-\eta$ using $\toh\left(\frac{l}{\sqrt{\delta}}\ln\left(\frac{1}{\eta}\right)\ln\left(\frac{1}{\epsilon}\right)\right)$ applications of $W_\beta$ and $W_\beta^{-1}$ at adaptively chosen values of $\beta$.
\end{proposition}

\begin{proof}
    Let $Z(\beta)=\sum_x\mu(x)e^{-\beta H(x)}$ be the partition function of the original Gibbs family. To apply \Cref{thm:q-adapt-anneal} in reverse, define the reparameterized partition function $\widehat{Z}(\beta):=Z(\beta^\star-\beta)$ for $\beta\in[0,\beta^\star]$. Equivalently, $\widehat{Z}(\beta)=\sum_x\mu(x)e^{-\beta^\star H(x)}e^{\beta H(x)}$. Thus, the corresponding Gibbs distribution is $\widehat{\pi}_\beta(x)=\pi_{\beta^\star-\beta}(x)$, so that $\widehat{\pi}_0=\pi_{\beta^\star}$ and $\widehat{\pi}_{\beta^\star}=\pi_0$. Since $Z(\beta)$ is log-convex, so is $\widehat{Z}(\beta)$, and the schedule length and construction argument applies also to this increasing log-convex partition function~\cite{HW20c,CH23c}. The assumption gives an $\epsilon$-approximate preparation of $\ket{\widehat{\pi}_0}=\ket{\pi_{\beta^\star}}$. Applying \Cref{thm:q-adapt-anneal} to $\widehat{Z}(\beta)$ therefore yields the desired schedule, which corresponds to the original schedule in decreasing order, from $\beta^\star$ to $0$.
\end{proof}

\begin{proposition}[\sc Computing a Slowly varying schedule]\label{prop:compute-adapt}
    Let $G=(V,E,w)$ be a weighted graph with $n$ vertices and $m$ edges. Fix an admissible spanning tree $S$, let $\beta^{\star}:=\ln(16n^{n-2}/\epsilon^2)$ for $\epsilon>0$, and let $\{G_\beta\}_{\beta\in[0,\beta^{\star}]}$ be the $\beta$-scaled graphs from \Cref{def:beta-scale}, with spanning tree distributions $\{\pi_\beta\}_{\beta}$. 

    Suppose the admissible spanning trees $S$ is stored explicitly together with its edge weights, then there is a quantum algorithm that, for any $\eta>0$, computes an $e^2$-slowly varying schedule $0=\beta_0<\cdots<\beta_l=\beta^{\star}$ of length $l \leq \toh(\sqrt n)$ with $\beta_0=0$, $\beta_l=\beta^{\star}$ with probability at least $1-\eta$. The algorithm uses $\toh(n\cdot \ln(1/\eta)\ln(1/\epsilon))$ applications of the operators $W_\beta$ obtained by applying \Cref{prop:qw-operator} to $G_{\beta}$.
\end{proposition}

\begin{proof}
    By \Cref{lem:log-convexity-partition}, $Z_S(\beta)$ is log-convex in $\beta$. The state $\ket{S}$ can be easily prepared since $S$ is known explicitly, and by \Cref{lem:pointmass}, $\norm{\ket{\pi_{\beta^{\star}}^S}-\ket{\pi_\infty^S}}_2\leq\epsilon/2$. Thus, by \Cref{prop:q-adapt-anneal-reverse}, there exists an $e^2$-slowly varying schedule from $0$ to $\beta^{\star}$ of length $l\leq\toh(\sqrt{Z_S(0)/Z_S(\beta^{\star})})\leq\toh(\sqrt{\pi_G(S)})\leq\toh(\sqrt{n})$, where the second inequality follows from $Z_S(\beta^{\star})\geq w_G(S)$ and the last from the fact that an $n$-vertex graph has at most $n^{n-2}$ spanning trees. By \Cref{lem:margaw-walk-gap}, the corresponding Markov chain $\mcl{M}_{\beta}$ has spectral gap at least $\wtd{\Omega}(1/n)$, while \Cref{prop:qw-operator} gives efficient implementations of the corresponding quantum walk operators $W_{\beta}$. Hence, the schedule can be computed using $\toh(\sqrt{n}\cdot\sqrt{n}\cdot\ln(1/\eta)\ln(1/\epsilon))=\toh(n\ln(1/\eta)\ln(1/\epsilon))$ applications of $W_\beta$ and $W_\beta^{-1}$.
\end{proof}

  \subsection{Time complexity}
  \label{Sec:time-complexity}
  For completeness, we briefly sketch the time complexity of \Cref{alg:q-sample}, in the model where the algorithm has QRAM access to classical data it has computed and stored, so each coherent lookup into a stored table costs $\operatorname{polylog}(n)$ time. Write $T_W$ for the time to implement one application of the walk operator $W_\beta$ (or its inverse) from \Cref{prop:qw-operator}. In the preprocessing phase, the admissible-tree computation and the reusable effective-resistance data structure $\mcl H$ (sparsification) take time $\toh(\sqrt{mn})$, matching their query complexity. The reusable marginal data structure $\margds$ is built by repeated Grover search over the edges, so it also takes time $\toh(\sqrt{mn/\rho})$, matching its query complexity. Converting $S$ to canonical form and reading the vertex degrees take $\toh(n)$ time. For each temperature $\beta$ at which we need $W_\beta$, we construct the effective-resistance oracle $O_{\mcl R}^{G_\beta}$ from $\mcl H$ in $\toh(n)$ time, and each subsequent evaluation of the oracle takes $\toh(1)$ time. Finally, computing the slowly-varying schedule uses $\toh(n)$ applications of the operators $W_\beta$ (\Cref{prop:compute-adapt}), for a time of $\toh(n\cdot T_W)$, plus the time to construct the effective-resistance oracles at the temperatures visited, namely $\toh(n\cdot l)$ where $l$ is the length of the slowly-varying schedule which is $\toh(\sqrt n)$. Thus it takes at most $\toh(n^{3/2})$.

We next bound $T_W$. Implementing $W_\beta$ requires a marginal-state preparation for the up step and a cycle-sampling procedure for the down step. Given $\margds$ and the effective-resistance oracle in QRAM, the marginal state $\ket{\mu_{G_\beta}}$ is prepared using $\toh(\sqrt{m\rho/n})$ queries (\Cref{prop:computable-heavy-envelope}). Each query is a stored-table lookup, and the state is loaded with the tree-based Grover--Rudolph method, so each costs $O(\log m)$ time, and the preparation takes $\toh(\sqrt{m\rho/n})$ time in total. For the down step, adding the sampled edge to the current tree $T$ creates a unique cycle. We find it by a breadth-first search on $T$ in $O(n)$ time, and then sample an edge of the cycle from the required distribution. This is done by building a binary tree over the cycle edges whose internal nodes store partial sums of the (inverse) weights, which takes $O(n)$ time to build and $O(\log n)$ time to sample from. Hence
\[
    T_W=\toh\big(\sqrt{m\rho/n}+n\big).
\]

Combining these bounds, the preprocessing time is
\[
    \toh\Big(\sqrt{mn/\rho}+n\cdot T_W+n\cdot l\Big)
    =\toh\Big(\sqrt{mn/\rho}+\sqrt{mn\rho}+n^{3/2}\Big)
    =\toh\Big(\sqrt{mn/\rho}+n^2\Big),
\]
where the first equality uses $l\leq\toh(\sqrt{n})$ and the second uses $\rho\leq1$. The time to prepare one q-sample is $\toh(n)$ walk operators (there are $l=\toh(\sqrt n)$ transformations, each using $\toh(\sqrt n)$ walk operators) plus $\toh(l\cdot n)=\toh(n^{3/2})$ for constructing the effective-resistance oracles along the schedule, that is
\[
    \toh\big(n\cdot T_W\big)=\toh\big(\sqrt{mn\rho}+n^2\big)=\toh(n^2),
\]
where the last step uses $m\le n^2$, so that $\sqrt{mn\rho}\le n^{3/2}$. Similarly, a reflection about $\ket{\pi}$ takes $\toh(\sqrt n\cdot T_W)=\toh(\sqrt{m\rho}+n^{3/2})$ time. Note that the time is dominated by the $\toh(n)$ cost of the down step in every walk operator, not by the query cost. As a result, the time complexity of q-sampling is $\toh(n^2)$ for every $\rho$, and the tradeoff in $\rho$ is visible in the query complexity but not in the time complexity. Matching the query tradeoff in time would require implementing the down step faster than $\toh(n)$.

\section{Lower bound for $k$ q-samples}
  \label{Sec:lower-bound}
  Here we prove our lower bound, showing that the algorithm's query complexity is optimal, up to polylogarithmic factors.
\begin{figure}[ht!]
    \begin{center}
        \begin{tikzpicture}[
    every node/.style={font=\large},
    vertex/.style={circle, fill=black, inner sep=2pt},
    dummy/.style={circle, fill=gray!70, inner sep=2.5pt}
]

\draw[thick] (0,0) ellipse (4.8 and 0.75);
\draw[thick] (0,2.6) ellipse (3.7 and 0.75);


\node[vertex] (l1) at (-3.5,0) {};
\node at (-3.5,-0.35) {$u_1$};
\node[vertex] (l2) at (-2.3,0) {};
\node at (-2.3,-0.4) {$u_2$};
\node[vertex] (l3) at (-1.3,0) {};
\node at (-1.3,-0.4) {$u_3$};
\node at (-0.35,0) {$\cdots$};
\node[vertex] (li) at (0.5,0) {};
\node at (0.5,-0.4) {$u_i$};
\node at (1.5,0) {$\cdots$};
\node[vertex] (l4) at (2.3,0) {};
\node at (2.3,-0.4) {$u_{n-1}$};
\node[vertex] (l5) at (3.5,0) {};
\node at (3.5,-0.35) {$u_n$};

\node[dummy]  (r0) at (-2.5,2.5) {};
\node[above left=-1pt of r0] {$v_0$};
\node[vertex] (r1) at (-1.5,2.5) {};
\node[above left=-1pt of r1] {$v_1$};
\node at (-0.7,2.5) {$\cdots$};
\node[vertex] (r3) at (0,2.5) {};
\node[above left=-1pt of r3] {$v_j$};
\node at (0.8,2.5) {$\cdots$};
\node[vertex] (r4) at (1.5,2.5) {};
\node[vertex] (rq) at (2.6,2.5) {};
\node[above right=-1pt of rq] {$v_q$};

\node[vertex] (s) at (0,5) {};
\node[above=3pt of s] {$s$};

\foreach \r in {r0,r1,r3,r4,rq}
    \draw[thick, color=orange] (s) -- (\r);

\draw[thick,color=gray!70] (li) -- (r0);
\draw[thick,color=red] (li) -- (r1);
\draw[thick,color=red] (li) -- (r3);
\draw[thick,color=blue!70] (li) -- (r4);
\draw[thick,color=red] (li) -- (rq);

\draw[thick, color=gray!70] (l1) -- (r0);
\draw[thick, color=gray!70] (l2) -- (r0);
\draw[thick, color=gray!70] (l4) -- (r0);
\draw[thick, color=gray!70] (l3) -- (r0);
\draw[thick, color=gray!70] (l5) -- (r0);

\foreach \l in {l1,l2,l3,l4,l5}{
    \draw[thick,color=red]                (\l) -- ++(-0.28,0.95);
    \draw[thick,color=blue!70]  (\l) -- ++(0.00,1.00);
    \draw[thick,color=red]                (\l) -- ++(0.28,0.95);
    \draw[thick,color=blue!70]  (\l) -- ++(0.55,0.85);
}

\end{tikzpicture}
    \end{center}
    \caption{Pictorial representation of the graph construction for the lower bound. The $n$ vertices
    $u_1,\ldots,u_n$ encode $n$ independent $k$-threshold instances $x^{(1)},\ldots,x^{(n)}\in\{0,1\}^q$ on $q=\lceil 2m/n\rceil$ bits. For each $i\in[n]$ and $j\in[q]$, the edge $(u_i,v_j)$ has weight $x^{(i)}_j$: red edges have weight $1$ and correspond to $x^{(i)}_j=1$, while blue edges have
    weight $0$ and correspond to $x^{(i)}_j=0$. The $v_0$ (gray) is joined to every $u_i$ by an edge of weight $1$. The orange edges connect $s$ to every $v_j$, $0\leq j\leq q$, and all have the fixed weight
    $\mcl{W}=12m$.}
    \label{fig:lower-bound}
\end{figure}
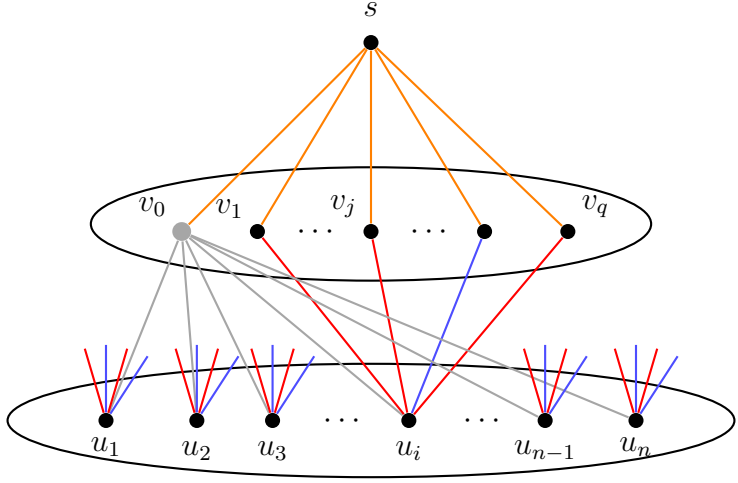

We first prove a lower bound for obtaining $k$ independent classical samples from the spanning tree distribution $\pi$.

\begin{lemma}\label{lem:sampling-lower-bound}
Let $G=(V,E,w)$ be a weighted graph with $n$ vertices, $m$ edges and $w\in\mathbb{R}_{\ge0}^E$,
let $k\leq m/n$ be a positive integer, and let $\epsilon\leq \frac{1}{12k}$.
Any quantum algorithm that, given access to $\orcl$, the adjacency-list oracle, outputs $k$ independent
samples from a distribution $\wtd\pi$ which satisfies $\norm{\wtd\pi-\pi}_{\rm TV}\leq\epsilon$, where $\pi$ is
the spanning tree distribution of $G$ must make $\Omega(\sqrt{kmn}/\ln(nk))$ queries to $\orcl$ in the worst-case.
\end{lemma}

\begin{proof}
Let $q=\lceil 2m/n\rceil$, so $q\geq 2k$ since $k\leq m/n$. We reduce from $n$ independent instances $x^{(1)},\dots,x^{(n)}\in\{0,1\}^q$ of the $k$-threshold problem, each defined on $q$ bits. In each instance $x^{(i)}$, the goal is to output $1$ iff $|x^{(i)}| \geq k$ i.e., given quantum query access to $x^{(i)}$ through the oracle $\mcl{O}_{x^{(i)}}\colon \ket{j,b}\mapsto \ket{j,b\oplus x^{(i)}_j}$. The graph $G$ (see \Cref{fig:lower-bound}) has vertices $s,v_0,\dots,v_q,u_1,\dots,u_n$ and edges $(s,v_j)$ of weight $\mcl W:=12m$ for $0\leq j\leq q$, $(u_i,v_0)$ of weight $1$ for $i\in[n]$, and $(u_i,v_j)$ of weight $x^{(i)}_j$ for $i\in[n]$, $j\in[q]$.

It has $n+q+2=O(n)$ vertices and $q+1+n(q+1)=\Theta(m)$ edges. The vertex degrees do not depend on $x$, and each query to $\orcl$ is can be simulated with at most one query to some $\mcl{O}_{x^{(i)}}$ (all other edges are fixed). The dummy vertex $v_0$ keeps $G$ connected even when some $x^{(i)}=0^q$, so $\pi$ is always defined. Let $M_i=\{0\}\cup\{j\in[q]:x^{(i)}_j=1\}$, so $|M_i|=|x^{(i)}|+1$.

We first show that most trees contain all edges joining  $s$. Call a spanning tree $T$ \emph{good} if $(s,v_j)\in T$ for all $0\leq j\leq q$. We know that the marginal of $(s,v_j)$ equals the leverage score, 
\[
    \Pr_{T\sim \pi}[(s,v_j)\in T]= \ell((s,v_j)) = \mcl W\cdot\reff^G(s,v_j).
\]
For a fixed $v_j$, \emph{short} all its neighbors other than $s$ into $s$, i.e., identify these vertices with $s$. By Rayleigh's Shorting law~\cite{DS84b}, this can only decrease the effective resistance between $s$ and $v_j$. The resulting network consists of the edge $(s,v_j)$, with resistance $1/\mcl W$, in parallel with $d_j\leq n$ unit-resistance edges, where $d_j = \bigl|\{i\in[n] : x^{(i)}_j=1\}\bigr|$. Thus $\reff^G(s,v_j)\geq 1/(n+\mcl W)$, so $\Pr_{T \sim \pi}[(s,v_j)\notin T]\leq n/(n+\mcl W)\leq n/\mcl W$. A union bound over the $q+1$ edges, gives 
\[
    \Pr_{T \sim \pi}[T\text{ not good}]\leq \frac{(q+1)n}{\mcl W} \leq \frac{4m}{\mcl W} = \frac{1}{3}.
\]

Next, we show that a good tree reveals a uniformly random element of $M_i$ for each instance $x^{(i)}$. Recall that a spanning tree $T$ is sampled with probability proportional to its weight, so any tree containing a zero-weight edge has weight zero and occurs with probability zero. Thus, we only need to consider positive-weight edges.

Condition on a good spanning tree $T$. Since all edges $(s,v_j)$ are present, every $v_j$ is already connected to $s$. Hence, each vertex $u_i$ needs exactly one incident edge to be connected to the rest of the tree. It cannot have two such edges, since this would create a cycle through $s$. Therefore, $u_i$ has exactly one incident edge, say $(u_i,v_{j_i(T)})$. For this edge to have positive weight, we must have $j_i(T)\in M_i$.

Conversely, every choice of $j_i\in M_i$ gives a good spanning tree of the same weight. Therefore, conditioned on $T$ being good, $j_i(T)$ is uniformly distributed over $M_i$. If $T$ is good and $j_i(T)\neq 0$, we say that $T$ \emph{reveals} the index $j_i(T)$ for $x^{(i)}$; otherwise $T$ reveals no index for $x^{(i)}$. A revealed index $j$ satisfies $x^{(i)}_j=1$, since $j_i(T)\in M_i$.

For each instance $i$, let $T_1,T_2,\ldots$ be independent samples from $\wtd{\pi}$, and let $D_i(t)$ be the set of indices revealed for $x^{(i)}$ by $T_1,\ldots,T_t$. Then $D_i(t)\subseteq D_i(t+1)$ and $|D_i(t)|\leq|x^{(i)}|$.

Suppose $|x^{(i)}|\geq k$. Fix $t$ and an outcome of $T_1,\ldots,T_t$ with $|D_i(t)|<k$, and let $\new$ be the event that $T_{t+1}$ reveals an index outside $D_i(t)$. Since $T_{t+1}\sim\wtd{\pi}$ is independent of $T_1,\ldots,T_t$, the probability $\Pr[\new\mid T_1,\ldots,T_t]$ is the probability that a fresh sample from $\wtd{\pi}$ reveals an index outside the set $D_i(t)$. For a tree $T\sim\pi$, $T$ is good with probability at least $2/3$, and conditioned on being good, $j_i(T)$ is uniform over $M_i$, which has $|x^{(i)}|+1$ elements, exactly $|x^{(i)}|-|D_i(t)|$ of which are new indices. Hence the probability that $T\sim\pi$ reveals an index outside $D_i(t)$ is at least
\[
\frac{2}{3}\cdot\frac{|x^{(i)}|-|D_i(t)|}{|x^{(i)}|+1}
=\frac{2}{3}\left(1-\frac{|D_i(t)|+1}{|x^{(i)}|+1}\right)
\geq\frac{2}{3}\left(1-\frac{|D_i(t)|+1}{k+1}\right)
=\frac{2}{3}\cdot\frac{k-|D_i(t)|}{k+1}.
\]
Our samples, however, come from $\wtd{\pi}$, and $\norm{\wtd{\pi}-\pi}_{\rm TV}\leq\epsilon$ changes the probability of any event by at most $\epsilon$, and $\epsilon\leq 1/(12k)\leq (k-|D_i(t)|)/(12k)$ and $k+1\leq 2k$,
\[
\Pr[\new\mid T_1,\ldots,T_t]\;\geq\;\frac{2(k-|D_i(t)|)}{3(k+1)}-\frac{k-|D_i(t)|}{12k}\;\geq\;\frac{k-|D_i(t)|}{4k}.
\]

Let $Y_t:=\max\{k-|D_i(t)|,0\}$ be the number of indices still missing, so that $|D_i(t)|<k$ if and only if $Y_t\geq 1$. We claim that $\E[Y_{t+1}\mid T_1,\ldots,T_t]\leq\rho\,Y_t$. If $Y_t=0$, then $Y_{t+1}=0$. Otherwise, since $D_i(t)\subseteq D_i(t+1)$, we always have $Y_{t+1}\leq Y_t$, and on $\new$ we have $Y_{t+1}\leq Y_t-1$. Thus $Y_{t+1}\leq Y_t-\mathds{1}_{\new}$, and by the bound above $\Pr[\new\mid T_1,\ldots,T_t]\geq Y_t/(4k)$, so $\E[Y_{t+1}\mid T_1,\ldots,T_t]\leq Y_t-Y_t/(4k)=\left(1-\frac{1}{4k}\right)\,Y_t$. Averaging over all outcomes of $T_1,\ldots,T_t$ gives $\E[Y_{t+1}]\leq\left(1-\frac{1}{4k}\right)\,\E[Y_t]$, and since $Y_0=k$, induction on $t$ gives $\E[Y_r]\leq k\left(1-\frac{1}{4k}\right)^{r}$. By Markov's inequality, and choosing $r=k\lceil 8\ln(nk)\rceil\geq 8k\ln(nk)$,
\[
\Pr_{(T_1,\ldots,T_r)\sim\wtd{\pi}^{\otimes r}}[|D_i(r)|<k]=\Pr[Y_r\geq 1]\leq\E[Y_r]\leq k\left(1-\tfrac{1}{4k}\right)^{r}\leq k\,e^{-r/(4k)}\leq k(nk)^{-2}.
\]

We now describe an algorithm for the $k$-threshold problem on $n$ independent instances, given quantum query access to $x^{(1)},\ldots,x^{(n)}$. It runs the assumed sampling algorithm independently $\lceil 8\ln(nk)\rceil$ times, answering each query to $\orcl$ with at most one query to some $\mcl{O}_{x^{(i)}}$. This gives $r=k\lceil 8\ln(nk)\rceil$ independent samples $T_1,\ldots,T_r$ from $\wtd{\pi}$. From these trees alone, without further queries, it computes $D_i(r)$ for every $i\in[n]$ and outputs $(b_1,\ldots,b_n)$, where $b_i=1$ if and only if $|D_i(r)|\geq k$. If the sampling algorithm makes $Q$ queries to $\orcl$, this algorithm makes $O(Q\ln(nk))$ queries.

The algorithm succeeds if $b_i=1$ exactly when $|x^{(i)}|\geq k$ for all $i\in[n]$ and it has one-sided error if every output $b_i=1$ is correct, that is, if $b_i=1$ only when $|x^{(i)}|\geq k$. Our algorithm has one-sided error: if $|x^{(i)}|<k$, then $|D_i(r)|\leq|x^{(i)}|<k$, since every revealed index $j$ satisfies $x^{(i)}_j=1$, so $b_i=0$ always. Hence it can fail only on an instance with $|x^{(i)}|\geq k$, where $b_i=0$ with probability at most $k(nk)^{-2}$ by the bound above. A union bound over at most $n$ such instances shows that it succeeds with probability at least $1-nk\cdot(nk)^{-2}=1-1/(nk)\geq 1/2$.

Since $k\leq q/2$, the bounded-error quantum query complexity of the $k$-threshold function on $q$ bits is $\Theta(\sqrt{kq})$~\cite{BBC+01}. By the direct product theorem for one-sided-error algorithms~\cite[Theorem~9]{ASdW07}, any one-sided-error quantum algorithm for $n$ independent instances of the $k$-threshold problem that makes at most $\alpha n\sqrt{kq}$ queries, for a sufficiently small constant $\alpha>0$, succeeds with probability at most $2^{-\Omega(nk)}<1/2$. Our algorithm succeeds with probability at least $1/2$, so it makes $\Omega(n\sqrt{kq})=\Omega(\sqrt{kmn})$ queries. Hence $Q\ln(nk)=\Omega(\sqrt{kmn})$, i.e.,
\[
Q=\Omega\left(\frac{\sqrt{kmn}}{\ln(nk)}\right).
\]
\end{proof}

Since measuring a q-sample yields a classical sample, the bound extends to q-samples.
\begin{theorem}\label{thm:lower-bound}
    Let $\epsilon\leq 1/(12k)$, let $n\leq m$ with $n$ sufficiently large, and let $k\leq m/n$. Any algorithm that, given adjacency-list access $\orcl$ to a weighted graph $G=(V,E,w)$ with $n$
    vertices, $m$ edges and $w\in\mathbb{R}_{\ge0}^E$, prepares $k$ independent copies of $\ket{\wtd\pi}$ with $\norm{\ket{\wtd\pi}-\ket{\pi}}_2\leq\epsilon$ makes $\Omega(\sqrt{kmn}/\ln(nk))$ queries to $\orcl$ in the worst case.
\end{theorem}
\begin{proof}
    Since $\norm{\ket{\wtd\pi}-\ket{\pi}}_{\rm TV} \leq \norm{\ket{\wtd\pi}-\ket{\pi}}_2\leq\epsilon$.
    Measuring in the computational basis (the basis of spanning trees) yields a classical sample of spanning tree from the distribution $\ket{\wtd\pi}$ which is within TV distance $\epsilon$ of $\pi$. Measuring each of the $k$ independent copies gives $k$ independent classical samples from it, at no
    extra query cost, so \Cref{lem:sampling-lower-bound} applies.
\end{proof}

We note that the restriction $k\leq m/n$ is necessary: reading the whole adjacency list takes $O(m)$ queries
and afterwards any number of samples (or copies of $\ket{\pi}$) can be produced without further
queries, so no bound above $O(m)$ is possible, and $\sqrt{kmn}>m$ exactly when $k>m/n$.
For $k>m/n$, an algorithm producing $k$ samples also produces $k'=\lfloor m/n\rfloor$ samples, so
\Cref{lem:sampling-lower-bound} with $k'$ gives $\Omega(m/\ln m)$. Hence the query complexity is
$\wtd\Theta(\min\{\sqrt{kmn},\,m\})$ for all $k$.

\newpage
\section*{Acknowledgements}
This work was supported by the Maison du Quantique de Nouvelle-Aquitaine ``HybQuant'', as part of the HQI initiative and France 2030, under the French National Research Agency (ANR) grant ANR-22-PNCQ-0002. It was also supported by the PEPR integrated project EPiQ under ANR grant ANR-22-PETQ-0007. We acknowledge the use of AI-based tools solely for language editing and proofreading.

\printbibliography[heading=bibintoc]

\appendix

\section{Useful quantum sub-routines}
\label{app:sec-q-subroutines}
\begin{lemma}[\sc Repeated Grover search, \cite{AdW20c,Gro97j}]\label{lem:repeated-grover}
    Let $f : [N] \rightarrow  \{0, 1\}$  be a function that
    marks a set of elements $S = \{i \in [N] \  | \ f(i) = 1\}$. Then for any $\eta >0$, there is a quantum algorithm that finds $S$ with probability at least $1-\eta$ in $\toh(\sqrt{N \cdot|S|})$ elementary operations and queries to $f$, and uses $O(\log{N})$ qubits and a QRAM of $\toh(|S|)$ bits.
\end{lemma}

We use Grover's $\frac{\pi}{3}$-amplitude amplification (fixed point search, \cite{Gro05j}) as in \cite{WA08j} to move between states in the cooling schedule.

\begin{lemma}[\sc Fixed-point amplitude amplification; adapted from {\cite[Corollary 1]{WA08j}}]\label{lem:fxd-pnt-amp}
    Let $\ket{\psi}$ and $\ket{\phi}$ be arbitrary quantum states in $\mathbb{C}^d$ with $\abs{\braket{\psi}{\phi}} \geq B$. Denote by $\Pi_{\psi}$ the projection on the subspace spanned by $\ket{\psi}$ and by $\Pi_{\psi}^{\perp}$ the projection onto the orthogonal subspace. Let $\omega=e^{\frac{\pi}{3} i}$. Define the unitaries $R_0 =\omega \Pi_{\psi}+\Pi_{\psi}^{\perp}$ and $ R_{1} =\omega \Pi_{\phi}+\Pi_{\phi}^{\perp}$. Given the state $\ket{\psi}$, we can prepare a state $\ket{\widetilde{\phi}}$ such that $\norm{\ket{\widetilde{\phi}}-\ket{\phi}}_2 \leq \epsilon_1$, for any $\epsilon_1>0$, by invoking the unitaries from $R_0$ and $R_1$ no more than
    \[
        L=\frac{12\ln\left(2/ \epsilon_1\right)}{\ln(1 /(1-B))}
    \]
    times.
\end{lemma}

\begin{lemma}[\sc Q-sample transformation from quantum walks; adapted from {\cite[Corollary~1, Corollary~2]{WA08j}}]\label{lem:q-sample-transform-walk}
Let $W_\psi$ and $W_\phi$ be quantum walk operators whose unique $1$-eigenvectors are $\ket{\psi}$ and $\ket{\phi}$, respectively, and suppose their phase gaps are at least $\gamma_{\rm ph}$. Assume $\abs{\braket{\psi}{\phi}}\ge B$. Then, for any $\epsilon>0$, given $\ket{\psi}$, there is a quantum circuit that outputs a state $\ket{\widetilde\phi}$ satisfying $\norm{\ket{\widetilde\phi}-\ket{\phi}}_2\le\epsilon$ using
\[
    \widetilde O\left(\frac{1}{\gamma_{\rm ph}}\cdot\frac{\log(1/\epsilon)}{\log(1/(1-B))}\right)
\]
controlled applications of $W_\psi,W_\phi$ and their inverses.
\end{lemma}

\begin{theorem}[\sc Quantum Maximum-Weight-Product Spanning Tree~\cite{AGJL25c}, \\adapted from~\cite{DHHM06j}]\label{thm:max-product}
    There exists a quantum algorithm that, for any $\eta>0$, given adjacency query access $\mathcal{O}_G$ to an undirected weighted graph $G=(V,E,w)$ with $n$ vertices and $m$ edges, whose edge weights are bounded by a constant $W$, outputs a maximum-weight-product spanning tree of $G$ with probability at least $1-\eta$. The algorithm uses $O(\sqrt{mn}\,\ln(1/\eta))$ queries to $\mathcal{O}_G$ and runs in $\tilde{O}(\sqrt{mn})$ time.
\end{theorem}

\section{Coined quantum walk}
\label{app:sec-coin-walk}
We restate and prove \Cref{thm:labelled-szegedy}. We use the notation and labelled-transition setup from \Cref{sec:quantum-walk}.

\begin{rtheorem}[\Cref{thm:labelled-szegedy} ({\sc Spectral mapping for labelled Szegedy walks})]
Suppose $P$ is a reversible, irreducible, and aperiodic Markov chain on $\mcl X$ with stationary distribution $\pi$. Then, on the busy subspace, the unique $1$-eigenvector of $W(P)$ is $\ket{\pi}\ket{0}$. Moreover,
\[
    \gamma_{\rm ph}(W(P))
    =\Omega(\sqrt{\delta_{\rm abs}(P)}).
\]
\end{rtheorem}

\begin{proof}
The proof follows the standard spectral analysis of Szegedy walks; see, e.g.,~\cite[Lecture~18]{Chi17p}.
We recall the notation from the main text. Let $\mcl A=\Span\{\ket{x}\ket{0}:x\in\mcl X\}$, $\mcl B=\mcl U^{\dagger} \mcl S U\mcl A$ and let $R_{\mcl A}$ and $R_{\mcl B}$ denote the corresponding reflections, so that
$W(P)=R_{\mcl B}R_{\mcl A}$. We identify $\mcl A$ with $\mathbb C^{\mcl X}$ via
$\ket{x}\ket{0}\leftrightarrow\ket{x}$.

For $x,y\in\mcl X$, the labelled-transition construction gives
\[
    \bra{y}\bra{0}\mcl U^{\dagger} \mcl S \mcl U \ket{x}\ket{0}
    =
    \sum_{a\in C_x:h_x(a)=y}
    \sqrt{p_x(a)p_y(a')}.
\]
For $x\neq y$, the unique-label assumption and reversibility therefore give
\[
    \bra{y}\bra{0}\mcl U^{\dagger} \mcl S \mcl U \ket{x}\ket{0}
    =
    \sqrt{P(x,y)P(y,x)}
    =
    \sqrt{\frac{\pi(x)}{\pi(y)}}P(x,y).
\]
For $x=y$, the self-loop labels contribute $P(x,x)$. Hence, using the discriminant matrix
$D(P):=D_\pi^{1/2}PD_\pi^{-1/2}$, $\Pi_{\mcl A}\,\mcl U^{\dagger} \mcl S \mcl U\,\Pi_{\mcl A}=D(P)$.

Let $\ket{v_\lambda}\in\mcl A$ be an eigenvector of $D(P)$ with eigenvalue $\lambda$. Since $\Pi_{\mcl B} = \mcl U^{\dagger} \mcl S \mcl U\,\Pi_{\mcl A}\,\mcl U^{\dagger} \mcl S \mcl U$, we obtain $\Pi_{\mcl B}\ket{v_\lambda} = \lambda\,\mcl U^{\dagger} \mcl S \mcl U\ket{v_\lambda}$. Writing $\ket{\wtd{v}_\lambda}:=R_{\mcl B}\ket{v_\lambda}$ gives
\[
    \ket{\wtd{v}_\lambda}
    =
    2\lambda \, \mcl U^{\dagger} \mcl S \mcl U\ket{v_\lambda}
    -\ket{v_\lambda},
    \qquad
    \Pi_{\mcl A}\ket{\wtd{v}_\lambda}
    =
    (2\lambda^2-1)\ket{v_\lambda}.
\]
Consequently, $\operatorname{Span}\{\ket{v_\lambda},\ket{\wtd{v}_\lambda}\}$ is invariant under $W(P)$, and in this basis
\[
    W(P)\big|_{\operatorname{Span}\{v_\lambda,\wtd{v}_\lambda\}}
    =
    \begin{pmatrix}
        0 & -1\\
        1 & 2(2\lambda^2-1)
    \end{pmatrix}.
\]
Its eigenvalues satisfy the quadratic equation $z^2-2(2\lambda^2-1)z+1=0$. Writing $\lambda=\cos\theta$ gives $z=e^{\pm 2i\theta}$. For the stationary eigenvector $\pi$, reversibility implies that
\[
    \mcl U\ket{\pi}\ket{0}
    =
    \sum_{x\in\mcl X}\sum_{a\in C_x}
    \sqrt{\pi(x)p_x(a)}\ket{x}\ket{a}
\]
is invariant under the flip-flop shift $\mcl S$. Hence $\mcl U^{\dagger} \mcl S \mcl U\ket{\pi}\ket{0} = \ket{\pi}\ket{0}$, so $\ket{\pi}\ket{0}\in\mcl A\cap\mcl B$ and is a $1$-eigenvector of $W(P)$. Since $P$ is irreducible and aperiodic, $1$ is a simple eigenvalue of $D(P)$ and all other eigenvalues satisfy $|\lambda|<1$. Thus this is the unique $1$-eigenvector on the busy subspace.

Finally, for every nontrivial eigenvalue $\lambda$ of $D(P)$, write $|\lambda|=\cos\theta$. Then $1-\abs{\lambda}\ge\delta$. Hence the smallest nonzero eigenphase $2\theta$ is $\Omega(\sqrt{\delta_{\rm abs}(P)})$. Therefore $\gamma_{\rm ph}(W(P))=\Omega(\sqrt{\delta_{\rm abs}(P)})$.
\end{proof}

\section{Cooling schedule}
\label{app:sec-cooling-sch}
\begin{rlemma}[\Cref{lem:log-convexity-partition} ({\rm Log-convexity of the partition function})]
Fix an admissible spanning tree $S$, and define $H_S(T):=|T\setminus S|$ and $Z_S(\beta):=\sum_{T\in\Omega}e^{-\beta H_S(T)}w(T)$. Then $Z_S(\beta)$ is log-convex in $\beta$.
\end{rlemma}

\begin{proof}
Fix $\beta_1,\beta_2\in\mathbb{R}$ and $\theta\in[0,1]$. Using $H_S(T)=|T\setminus S|$, we have
\[
    Z_S(\theta\beta_1+(1-\theta)\beta_2)=\sum_{T\in\Omega}\left(e^{-\beta_1H_S(T)}w(T)\right)^\theta\left(e^{-\beta_2H_S(T)}w(T)\right)^{1-\theta}.
\]
By H\"older's inequality with exponents $1/\theta$ and $1/(1-\theta)$,
\[
    Z_S(\theta\beta_1+(1-\theta)\beta_2)\leq\left(\sum_{T\in\Omega}e^{-\beta_1H_S(T)}w(T)\right)^\theta\left(\sum_{T\in\Omega}e^{-\beta_2H_S(T)}w(T)\right)^{1-\theta}=Z_S(\beta_1)^\theta Z_S(\beta_2)^{1-\theta}.
\]
Taking logarithms gives $\log Z_S(\theta\beta_1+(1-\theta)\beta_2)\leq\theta\log Z_S(\beta_1)+(1-\theta)\log Z_S(\beta_2)$, so $\log Z_S$ is convex. Hence $Z_S$ is log-convex.
\end{proof}

\section{Spectral sparsifiers}
\label{app:sec-spars}
\begin{proposition}[\sc Sum of Laplacians]\label{app:sum-lpcl}
    Let $G_1=(V,E_1,w_1)$ and $G_2=(V,E_2,w_2)$ be two weighted graphs on a common vertex set $V$. Then, for any $c\geq0$,
    \[
        L(G_1+c\cdot G_2)=L(G_1)+c\cdot L(G_2).
    \]
\end{proposition}

\begin{proof}
    Recall that $L(G)=\sum_{e\in E}w(e)b_eb_e^\top$.
    By \Cref{Def:del-scale}, the weight of every edge $e$ in
    $G_1+c\cdot G_2$ is $w_1(e)+c\,w_2(e)$, taking the weight to be zero
    when $e$ is absent from the corresponding graph. Therefore,
    \[
    L(G_1+c\cdot G_2)
    =
    \sum_e\bigl(w_1(e)+c\,w_2(e)\bigr)b_eb_e^\top
    =
    L(G_1)+c\cdot L(G_2).
    \]
\end{proof}

\begin{lemma}[\sc Sparsifier under disjoint edge addition]
\label{app:lem-add-spars}
    Let $G_1=(V,E_1,w_1)$ and $G_2=(V,E_2,w_2)$ be weighted graphs on the
    same vertex set, with $E_1\cap E_2=\emptyset$. If $\wtd H$ is an
    $\epsilon$-spectral sparsifier of $G_1$, then $\wtd H+G_2$ is an
    $\epsilon$-spectral sparsifier of $G_1+G_2$.
\end{lemma}

\begin{proof}
    Since $\wtd H$ is an $\epsilon$-spectral sparsifier of $G_1$, for every
    $x\in\mathbb R^{|V|}$,
    \[
    (1-\epsilon)x^\top L_{G_1}x
    \le x^\top L_{\wtd H}x
    \le (1+\epsilon)x^\top L_{G_1}x.
    \]
    Since $L_{G_2}\succeq0$,
    \[
    (1-\epsilon)x^\top L_{G_2}x
    \le x^\top L_{G_2}x
    \le (1+\epsilon)x^\top L_{G_2}x.
    \]
    Adding the two inequalities gives
    \[
    (1-\epsilon)x^\top(L_{G_1}+L_{G_2})x
    \le x^\top(L_{\wtd H}+L_{G_2})x
    \le (1+\epsilon)x^\top(L_{G_1}+L_{G_2})x.
    \]
    By \Cref{Def:del-scale,app:sum-lpcl}, $L_{G_1+G_2}=L_{G_1}+L_{G_2}$ and
    $L_{\wtd H+G_2}=L_{\wtd H}+L_{G_2}$. Hence $\wtd H+G_2$ is an
    $\epsilon$-spectral sparsifier of $G_1+G_2$.
\end{proof}

\begin{rlemma}[\Cref{lem:spars-reuse} ({\sc Sparsifier reusability})]
    Let $G=(V,E,w)$ be a weighted graph and fix $F\subseteq E$. Let $H$ be
    the subgraph of $G$ with edge set $E\setminus F$, and suppose that
    $\wtd H$ is an $\epsilon$-spectral sparsifier of $H$. Then, for every
    $c\geq0$, $F+c\cdot\wtd H$ is an $\epsilon$-spectral sparsifier of
    $G_{E\setminus F,c}$.
\end{rlemma}

\begin{proof}
    Since $\wtd H$ is an $\epsilon$-spectral sparsifier of $H$, scaling all
    edge weights by any $c\geq0$ shows that $c\cdot\wtd H$ is an
    $\epsilon$-spectral sparsifier of $c\cdot H$. Moreover, by
    \Cref{Def:del-scale}, $G_{E\setminus F,c}=F+c\cdot H$. Applying
    \Cref{app:lem-add-spars} to $c\cdot H$ and $F$ gives the claim.
\end{proof}

\end{document}